%% file: survey.tex
\documentclass{article}

\usepackage{hyperref,graphicx,amsthm,amsmath,amssymb}

\newtheorem{theorem}{Theorem}[section]
\newtheorem{definition}{Definition}[section]
\newtheorem{example}{Example}[section]
\newtheorem{lemma}{Lemma}[section]

\begin{document}

\title{Network modeling methods for precision medicine}

\author{

	Elio Nushi$^1$ \and 
	Victor-Bogdan Popescu$^2$ \and
	Jose Angel Sanchez Martin$^3$ \and
    Sergiu Ivanov$^4$ \and
    Eugen Czeizler$^{2,5}$ \and
    Ion Petre$^{5,6,*}$

}

\footnotetext[1]{
	Department of Computer Science, University of Helsinki, Helsinki, Finland
}

\footnotetext[2]{
	Department of Information Technologies, \r{A}bo Akademi University, Turku, Finland
}

\footnotetext[3]{
     Department of Computer Science, Technical University of Madrid, Madrid, Spain
}

\footnotetext[4]{
	IBISC Laboratory, Universit\'{e} Paris-Saclay, Universit\'{e} \'{E}vry, \'{E}vry, France
}

\footnotetext[5]{
	Department of Bioinformatics, National Institute of Research and Development for Biological Sciences, Bucharest, Romania
}

\footnotetext[6]{
    Department of Mathematics and Statistics, University of Turku, Finland
}

\renewcommand*{\thefootnote}{\fnsymbol{footnote}}

\footnotetext[1]{
    Address for correspondence: \texttt{ion.petre@utu.fi}
}

\renewcommand*{\thefootnote}{\arabic{footnote}}

\date{
	April 2021
}

\maketitle

\begin{abstract}

	We discuss in this survey several network modeling methods and their applicability to precision medicine. We review several network centrality methods (degree centrality, closeness centrality, eccentricity centrality,  be\-tween\-ness centrality, and eigenvector-based prestige) and two systems con\-tro\-la\-bi\-li\-ty methods (minimum dominating sets and network structural controllability). We demonstrate their applicability to precision medicine on three multiple myeloma patient disease networks. Each network consists of protein-protein interactions built around a specific patient's mutated genes, around the targets of the drugs used in the standard of care in multiple myeloma, and around multiple myeloma-specific essential genes. For each network we demonstrate how the network methods we discuss can be used to identify personalized, targeted drug combinations uniquely suited to that patient.

	\textbf{Keywords}: network medicine, computational modeling, precision me\-di\-cine, drug repurposing, centrality measures, graph theory, systems controllability, 

\end{abstract}

\input{section-introduction.tex}

\input{section-methods.tex}

\input{section-applications.tex}

\input{section-conclusion.tex}

\bibliographystyle{plain}
\bibliography{bibliography}

\end{document}

%% file: section-introduction.tex
\section{Introduction}
\label{section-introduction}

Network medicine is a promising recent approach in which the goal is to analyze the dysregulation of a disease through its specific molecular interactions (\cite{Saqi:2016aa, Tian:2012aa}). The key analytic power of this approach is that knowledge about disease-drivers and specific pathway deregulations can be combined with mechanistic knowledge of drug mechanisms to identify optimal drug combinations, and do this dynamically throughout the evolution of the disease.

An exciting aspect of this approach is that it can, in principle, be applied in a personalized way, taking into account patient-specific aspects such as co-morbidities, previous treatments, and the patient's own molecular data (such as her mutations, gene expression anomalies, corrupted signaling pathways). Therefore, a disease is seen as part of a patient's own molecular and clinical context, through the cumulative effect of various deregulations and anomalies. Also, drug therapies are seen as external interventions aiming to compensate for the effects of these anomalies in the patient-specific disease network. The focus is on identifying tailored drug combinations uniquely suited to that patient's disease network, in the current step of her disease progression.

This survey introduces several network modeling methods and demonstrates their potential applicability in personalized medicine. We survey several network centrality measures, aiming to identify parts of the network that are unusual in the context of its topology; we discuss their definitions and the intuition of their significance. We also discuss two systems controllability methods: network controllability and maximum dominating sets. The aim of these methods is to identify efficient interventions to change the network's configuration, and in principle to change from a setup associated with disease to one associated with a healthy state.

We demonstrate how these network modeling methods can be used in precision medicine for identifying targeted, personalized drug combinations. Our case study is multiple myeloma, an incurable cancer of the blood. We analyze a dataset consisting of the genetic mutations of three different patients. For each of them we construct their own personalized protein-protein directed interaction networks, and we analyze them with some of the methods in this survey to extract personalized predictions of optimal drug combinations. We compare these results with the standard therapy lines in multiple myeloma.

We also include a brief discussion on the availability of software tools supporting this line of research.

The survey is written in a tutorial style to facilitate the adoption of these methods.

%% file: section-methods.tex
\section{Network modeling methods}
\label{section-methods}

We discuss in this section a number of network analysis methods: network centrality measures (including degree centralities, proximity centralities, path centralities and spectral centralities) and two systems controllability methods (network controllability and minimum dominating sets). We apply several of these methods to a medical case-study in  Section~\ref{section-applications}.

\input{subsection-methods-centrality}

\input{subsection-methods-controllability}

\input{subsection-methods-software}

%% file: subsection-methods-centrality.tex
\subsection{Network centrality methods}
\label{subsection-methods-centrality}

Real-world networks often include a large number of nodes and connections, but the importance of the nodes is generally not the same. The simplest way to measure the importance of a given node is to compute its degree: the number of incident edges. However, in many cases, more sophisticated approaches are required to produce meaningful measures of importance. In the most general sense, a centrality measure can be defined in the following way.

\begin{definition}[Centrality]
    Let $G = (V, E)$ be a network. A \emph{centrality measure} is any function $f : V \to \mathbb R$.
\end{definition}

This definition imposes no constraints on $f$, but most centrality measures take into account the structural properties of the network $G$ --- node connectivity, edge weights, etc. Based on which structural properties they take into account, centrality measures can be grouped into the following categories:

\begin{itemize}
    \item \emph{degree centralities:} measures based on the degree of a node;
    \item \emph{proximity centralities:} measures based on how close a node is to the other nodes in the network;
    \item \emph{path centralities:} measures based on the role the node plays in paths that traverse it;
    \item \emph{spectral centralities:} measures related to the algebraic properties of the adjacency matrix of the network (in particular its eigenvectors and eigenvalues).
\end{itemize}

While the majority of centralities focus on individual nodes, one can define measures focused on other structures: edges, subsets of nodes, etc. Most of these measures are straightforward derivations from node-based centralities (e.g., edge and group betweenness in~\cite{Brandes2008}), and we will not discuss them here. Furthermore, we only consider unweighted networks, and we focus on structural measures which do not take into account any possible dynamical states.

In the rest of this subsection, we discuss some well-known and often used centrality measures in detail. In particular, we give the intuitive motivations, the formal definitions, and the highlighted properties of the network. Furthermore, we give references to algorithms for computing the centrality measures and briefly discuss their time complexities. Finally, we describe \emph{network centrality indices} --- network-wide scores measuring the centralization of a network as a whole, and allowing for comparisons between networks.

While we aim at a comprehensive overview of the state of the art, we do not always provide a fully detailed discussion of all subjects. For in-depth treatment, we refer the reader to~\cite{Junker2008,Brandes2005,Newman2010}.

In the following sections we will frequently refer to undirected and directed star-topology networks. The $k$-node undirected star-topology network is $U_k^\star = (V, E)$ (also known as the full bipartite graph $K_{1,k}$), where $V = \{1, \dots, k\}$ and $E = \{\{i, 1\} \mid 2 \leq i \leq k\}$. The $k$-node directed star-topology network is $G_k^\star = (V, E')$, where $V$ is the same as in $U_k^\star$ and $E' = \{(i, 1) \mid 2 \leq i \leq k\}$. Note that $E$ consist of unordered pairs of vertices, while $E'$ consist of ordered pairs. We will also write $U^\star$ and $G^\star$ to refer to the general notion of undirected and directed star-topology networks respectively.

\subsubsection{Running example}
\label{sec:rex}

We will use the directed network in Figure~\ref{fig:rex} to illustrate the presented centrality measures and the related concepts. This is a scale-free random network generated using the Python library \emph{NetworkX}~\cite{networkx-article,networkx-url} with the following line of code: \texttt{networkx.scale\_free\_graph(10, alpha=.7, beta=.2, gamma=.1, seed=3)}.

\begin{figure}[htb]
    \begin{center}
        \includegraphics[width=0.9\textwidth]{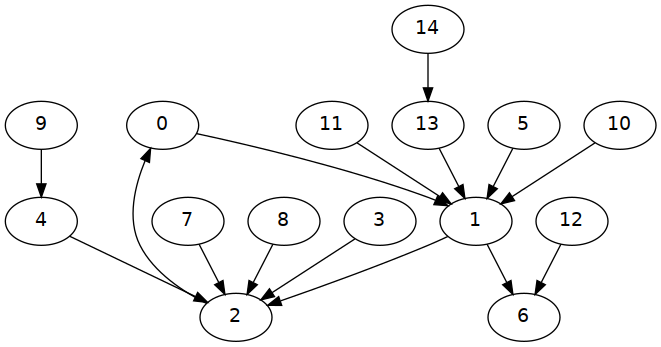}
    \end{center}
    \caption{The network serving as a running example to illustrate centralities and the related concepts.}
    \label{fig:rex}
\end{figure}

\subsubsection{Degree centralities}
\label{sec:degree-centralities}

Degree centrality was first introduced in~\cite{Shaw1954} to study the structure and behavior of groups of individuals in a society. It measures the importance of a node by directly counting the adjacent edges.

\begin{definition}[Degree centrality]
    \label{def:degree-centrality}

    Let $G = (V, E)$ be a network.

    \begin{itemize}
        \item If $G$ is undirected, then the \emph{degree centrality} is the function $\mathcal{C}_D : V \to \mathbb R$ assigning to every node its degree: $\mathcal{C}_D(v) = deg(v)$.
        \item If $G$ is directed, then the \emph{in-degree centrality} is the function $\mathcal{C}_D^- : V \to \mathbb R$, assigning to every node its in-degree: $\mathcal{C}_D^-(v) = deg^-(v)$. The \emph{out-degree centrality} $\mathcal{C}_D^+ : V \to \mathbb R$ assigns to every node its out-degree $\mathcal{C}_D^+(v) = deg^+(v)$. The (full) degree centrality is the function $\mathcal{C}_D(v) = \mathcal{C}_D^+(v) + \mathcal{C}_D^-(v)$.
    \end{itemize}

\end{definition}

It follows from these definitions that $\mathcal{C}_D(v)$ (or $\mathcal{C}_D^+(v)$ and $\mathcal{C}_D^-(v)$) is large when the node $v$ is adjacent to a high number of nodes. The extreme cases are $\mathcal{C}_D(v) = k-1$, in which $v$ is connected to all other nodes in a $k$-node network and $C(v) = 0$, when $v$ is isolated in the network.

The degree centrality is strongly related to the number of nodes in a network. As a trivial example, consider the degree centrality of the central node in the star-topology network $U_k^\star$: $\mathcal{C}_D^+(1) = k - 1$. The centrality of node $1$ is larger in larger networks with this topology, even though the intuitive idea of the importance of this node is essentially the same: it is connected to all other nodes of the network in both cases. Normalized degree centrality may be used to better capture the independence of the notion of centrality on the size of the network.

\begin{definition}[Normalized degree centrality]
    
    Let $G = (V, E)$ be a (either directed, or undirected) network. The \emph{normalized degree centrality} is the function $\widetilde{\mathcal{C}}_D : V \to \mathbb R$ defined as follows:

    \[
        \widetilde{\mathcal{C}}_D(v) = \frac{\mathcal{C}_D(v)}{|V| - 1}.
    \]

\end{definition}

Normalized degree centrality $\widetilde{\mathcal{C}}_D(v)$ therefore gives the ratio of the nodes adjacent to $v$, and it follows from the definition that $0 \leq \widetilde{\mathcal{C}}_D(v) \leq 1$. $\widetilde{\mathcal{C}}_D(v)$ can be thought of as the probability of $v$ to be connected to another node $w$ picked at random, an intuition that can be useful when generating random networks with a fixed degree distribution.

\begin{example}
  
    The following table gives the degree, in-degree, and out-degree centralities respectively for the example network in Figure~\ref{fig:rex}.

    \begin{center}
        \begin{tabular}{r|ccccccccccccccc}
            $v$ & 0 & 1 & 2 & 3 & 4 & 5 & 6 & 7 & 8 & 9 & 10 & 11 & 12 & 13 & 14 \\
            \hline
            $\mathcal{C}_D(v)$ & 2 & 7 & 6 & 1 & 2 & 1 & 2 & 1 & 1 & 1 & 1 & 1 & 1 & 2 & 1 \\
            $\mathcal{C}_D^-(v)$ & 1 & 5 & 5 & 0 & 1 & 0 & 2 & 0 & 0 & 0 & 0 & 0 & 0 & 1 & 0 \\
            $\mathcal{C}_D^+(v)$ & 1 & 2 & 1 & 1 & 1 & 1 & 0 & 1 & 1 & 1 & 1 & 1 & 1 & 1 & 1 \\
        \end{tabular}
    \end{center}

\end{example}

The complexity of computing the degree centrality of every individual node of a network depends linearly on the number of nodes and edges: $O(|V| + |E|)$ (e.g., \cite{Das2018}). Computing the centrality of any given node may be of the complexity $O(|V|)$ or $O(|E|)$, depending on the data structure used to represent the connections.

The degree centrality is a useful tool for identifying ``targets'' in a given network, but also for deciding which nodes may be discarded without impacting the quality of the model. This is common in network biology, see~\cite{Hahn2004}, \cite{Koschutzki2008}. However, degree centrality is a strongly local measure, mostly focusing on individual nodes and their immediate neighborhoods. In practice, this means that many nodes may often have close degree centralities, requiring finer measures to discern relevant features (e.g.~\cite{Kang2011}).

\subsubsection{Proximity centralities}
\label{sec:proximity-centralities}

In this section we discuss closeness, harmonic, and eccentricity centralities.

The degree centrality measure is a very straightforward approach to evaluate the importance of a node in a network. However, the nodes which are connected by a small number of edges to many others in a network are also important: their ``influence'' can reach many other nodes quickly. Counting the neighbors of a node clearly does not suffice to asses this kind of closeness to other nodes. The important aspect is rather having a small average distance to the other nodes in the network. This leads to the following definition of the closeness centrality measure as the reciprocal of farness~\cite{Bavelas1950,Leavitt1949,Sabidussi1966}.

\begin{definition}[Closeness centrality]

    Let $G = (V, E)$ be an undirected network. The \emph{closeness centrality} is the function $\mathcal{C}_C : V \to \mathbb R$ defined as follows:

    \[
        \mathcal{C}_C(v) = \frac{1}{\displaystyle \sum_{u \in V} d(u, v)}.
    \]

\end{definition}

As in the case of degree centralities, the size of the network has an impact on the closeness centrality of its nodes: the larger the network, the more paths it contains, and the lower closeness centralities tend to become. To make closeness centralities more uniform, \cite{Beauchamp1965} proposed the normalized version of this measure.

\begin{definition}[Normalized closeness centrality]
    Let $G = (V, E)$ be a connected undirected network. The \emph{normalized closeness centrality} is the function $\widetilde{\mathcal{C}}_C : V \to \mathbb R$ is defined as follows:

    \[
        \widetilde{\mathcal{C}}_C(v) = \frac{|V| - 1}{\displaystyle \sum_{u \in V} d(u, v)}.
    \]

\end{definition}

Normalized closeness centrality can be thought of as the inverse of the mean of the distances to $v$ from all other nodes. The bounds on the values of the normalized closeness centrality $\widetilde{\mathcal{C}}_C$ are the same as the bounds on the values of $\mathcal{C}_C$: $0 < \widetilde{\mathcal{C}}_C(v) \leq 1$, but $\widetilde{\mathcal{C}}_C(v)$ reaches its maximal value 1 for any node $v$ adjacent to all the other nodes of a given network.

A major drawback of closeness centrality is that it does not yield meaningful values on disconnected networks. Indeed, if no path connects nodes $u$ and $v$, by definition $d(u, v) = \infty$, and therefore a single isolated node would make the sums of distances in the definition of closeness infinite, and the closeness centralities themselves will all become 0. There are several different ways in which this can be addressed:

\begin{itemize}
    \item Restrict the notion of closeness centrality to strongly connected graphs (or strongly connected components of arbitrary graphs). This avoids the problem is having pairs of nodes $(u,v)$ whose distance is infinite, on the grounds of $v$ being unreachable from $u$.
    \item Restrict the sum of distances in the definition of closeness centrality to pairs of reachable nodes, addressing the same issue of infinite distances.
    \item Replace any infinite distances with a large enough constant, as proposed in \cite{Rochat2009} and~\cite{Csardi2006}.
\end{itemize}

Closeness centrality can also be defined for directed networks. To do this, we can consider in the definition either the distances $d(u,v)$ from all ancestors of $v$ to $v$, or distances $d(v,u)$ from $v$ to all its descendants. This is important when using the closeness centrality measure as a proxy for the notion of either a node that is reachable (and modifiable) from many directions, or that of a node that is influential in being able to reach many other nodes. One example of such a definition is the \emph{Lin index} \cite{Lin1976}. The software package \emph{NetworkX} computes the distances from the nodes which reach $v$~\cite{networkx-closeness}, and normalizes with respect to the number of these nodes. The difficulties with infinite distances for pairs of unreachable nodes persist also in the directed case, with possible solutions similar to those for the undirected case.

\begin{example}

    The following table gives the closeness centralities and the normalized closeness centralities, rounded to two digits after the decimal point, for the nodes of the example network in Figure~\ref{fig:rex}. The lines $C_C^u(v)$ and $\widetilde{C}_C^u(v)$ take this network to be undirected, meaning that the directed edges appearing in the figure can be traversed both ways. The calculations for directed networks were done with \emph{NetworkX}, using the approach explained above.

    \begin{center}
        \begin{tabular}{r|ccccccccc}
            $v$ & 0 & 1 & 2 & 3 & 4 & 5 & 6 & 7 \\
            \hline
            $\mathcal{C}_C(v)$ & 0.03 & 0.04 & 0.05 & 0.00 & 1.00 & 0.00 & 0.03 & 0.00 \\
            $\widetilde{\mathcal{C}}_C(v)$ & 0.40 & 0.48 & 0.60 & 0.00 & 1.00 & 0.00 & 0.36 & 0.00 \\
            \hline
            $\mathcal{C}_C^u(v)$ & 0.03 & 0.05 & 0.04 & 0.03 & 0.03 & 0.03 & 0.03 & 0.03 \\
            $\widetilde{\mathcal{C}}_C^u(v)$ & 0.48 & 0.64 & 0.58 & 0.38 & 0.40 & 0.40 & 0.42 & 0.38 \\
        \end{tabular}
    \end{center}

    \begin{center}
        \begin{tabular}{r|cccccccc}
            $v$ & 8 & 9 & 10 & 11 & 12 & 13 & 14 \\
            \hline
            $\mathcal{C}_C(v)$ & 0.00 & 0.00 & 0.00 & 0.00 & 0.00 & 1.00 & 0.00 \\
            $\widetilde{\mathcal{C}}_C(v)$ & 0.00 & 0.00 & 0.00 & 0.00 & 0.00 & 1.00 & 0.00 \\
            \hline
            $\mathcal{C}_C^u(v)$ & 0.03 & 0.02 & 0.03 & 0.03 & 0.02 & 0.03 & 0.02 \\
            $\widetilde{\mathcal{C}}_C^u(v)$ & 0.38 & 0.29 & 0.40 & 0.40 & 0.30 & 0.42 & 0.30 \\
        \end{tabular}
    \end{center}

\end{example}

Computing the closeness centrality of a node $v$ requires finding the shortest paths to $v$ from all other nodes of the network. Since constructing all shortest path to one particular node $v$ is of complexity $O(|E|)$, using e.g. a breadth-first search, computing the closeness centrality for all the nodes within a network is of complexity $O(|V| \cdot |E|)$. This means that computing the exact value of closeness centrality is impractical for many biological networks, which often contain thousands of nodes and tens of thousands of connections. It turns out that, in practice, one often only needs the first $k$ nodes with the highest closeness centrality, without requiring the actual centrality values. Such rankings can be computed in reasonable time even for very large networks, see, e.g.,~\cite{Bergamini2019}.

Due to its non-locality, closeness centrality is a finer tool for structural network analysis than degree centrality. For example, closeness fares better in identifying influential groups of nodes which may not individually have  high degree centrality. Closeness centrality has been show to perform particularly well in
biological network analysis. For example, \cite{Ma2003} shows that a slightly modified closeness measure allows for associating 8 of the top 10 metabolites of the metabolic network of \emph{E.~coli} with the glycolysis and citric acid cycle pathways.

A modification of closeness centrality, addressing the difficulty of infinite distances, consists in swapping the summation out of the denominator, effectively transforming what is an inverse arithmetic mean in normalized closeness centrality into inverse \emph{harmonic} mean. This new centrality measure was first discussed in~\cite{Marchiori2000}, then independently introduced in~\cite{Dekker2005} under the name ``valued centrality'', and finally gained its current name of harmonic centrality in~\cite{Rochat2009}.

\begin{definition}[Harmonic centrality]

    Let $G = (V, E)$ be a (either directed, or undirected) network. The \emph{harmonic centrality} is the function $\mathcal{C}_H : V \to \mathbb R$ defined as follows:
    
    \[
        \mathcal{C}_H(v) = \sum_{u \in V \setminus \{v\}} \frac{1}{d(u,v)}.
    \]

\end{definition}

Note that all nodes $u$ which are not connected to $v$ do not contribute to $\mathcal{C}_H(v)$, because $d(u,v) = \infty$, meaning that $1/d(u,v) = 0$.

As with closeness centrality, the same definition of harmonic centrality can be used for directed networks, in which case the order of nodes in the denominator $d(u,v)$ becomes important. If $v$ has in-degree 0, $C_H(v) = 0$ by direct computation of the formula in the definition.

To avoid an increase in harmonic centrality only due to the increase in the size of the network, one defines the normalized version.

\begin{definition}[Normalized harmonic centrality]

    Let $G = (V, E)$ be a (either directed, or undirected) network. The \emph{normalized harmonic centrality} is the function $\widetilde{\mathcal{C}}_H : V \to \mathbb R$ defined as follows:

    \[
        \widetilde{\mathcal{C}}_H(v) = \frac{1}{|V|-1}\sum_{u \in V \setminus \{v\}} \frac{1}{d(u,v)}.
    \]

\end{definition}

It follows that $0 \leq \widetilde{\mathcal{C}}_H(v) \leq 1$ (and $0 \leq \mathcal{C}_H(v) \leq |V|-1)$ both in directed and undirected networks. $\widetilde{\mathcal{C}}_H(v) = 0$ for isolated vertices, while $\widetilde{\mathcal{C}}_H(v) = 1$ ($\mathcal{C}_H(v) = |V| - 1$) for the center of a star network, in both directed and undirected cases, because both in $U_k^\star$ and $G_k^\star$ every node is connected to the central node 1.

Note that central nodes in large connected components will have greater values of harmonic centrality than central nodes in small connected components. Furthermore, nodes in disconnected networks will tend to have lower harmonic centralities than nodes in connected networks.

\begin{example}

    The following table gives the harmonic centralities and the normalized harmonic centralities, rounded to two digits after the decimal point, for the nodes of the example network in Figure~\ref{fig:rex}. The lines $\mathcal{C}_H^u(v)$ and $\widetilde{\mathcal{C}}_H^u(v)$ take this network to be undirected, meaning that the directed edges appearing in the figure can be traversed both ways.

    \begin{center}
        \begin{tabular}{r|ccccccccc}
            $v$ & 0 & 1 & 2 & 3 & 4 & 5 & 6 & 7 \\
            \hline
            $\mathcal{C}_H(v)$ & 5.42 & 7.58 & 8.33 & 0.00 & 1.00 & 0.00 & 6.37 & 0.00 \\
            $\widetilde{C}_H(v)$ & 0.39 & 0.54 & 0.60 & 0.00 & 0.07 & 0.00 & 0.45 & 0.00 \\
            \hline
            $\mathcal{C}_H^u(v)$ & 7.50 & 10.33 & 9.67 & 6.00 & 6.67 & 6.25 & 6.92 & 6.00 \\
            $\widetilde{\mathcal{C}}_H^u(v)$ & 0.54 & 0.74 & 0.69 & 0.43 & 0.48 & 0.45 & 0.49 & 0.43 \\
        \end{tabular}
    \end{center}

    \begin{center}
        \begin{tabular}{r|cccccccc}
            $v$ & 8 & 9 & 10 & 11 & 12 & 13 & 14 \\
            \hline
            $\mathcal{C}_H(v)$ & 0.00 & 0.00 & 0.00 & 0.00 & 0.00 & 1.00 & 0.00 \\
            $\widetilde{C}_H(v)$ & 0.00 & 0.00 & 0.00 & 0.00 & 0.00 & 0.07 & 0.00 \\
            \hline
            $\mathcal{C}_H^u(v)$ & 6.00 & 4.82 & 6.25 & 6.25 & 4.95 & 6.92 & 4.95 \\
            $\widetilde{\mathcal{C}}_H^u(v)$ & 0.43 & 0.34 & 0.45 & 0.45 & 0.35 & 0.49 & 0.35 \\
        \end{tabular}
    \end{center}

\end{example}

The computational complexities related to the harmonic centrality are the same as those of the closeness centrality, because of the similarities in the definitions of the two measures. Since finding the shortest path is of complexity $O(|E|)$, computing the harmonic centrality in a given directed or undirected network $G = (V, E)$ is of complexity $O(|V| \cdot |E|)$. For very large networks, approximate calculation strategies can be used, or alternatively the direct computation of centrality can be replaced by finding the top $k$ nodes with the highest centrality value, similarly to~\cite{Bergamini2019}.

Like closeness centrality, harmonic centrality has great potential for analysis of biological networks, because it captures the intuition of the influence of a node decaying with the distance, while also naturally handling disconnected networks. Online resources for systems biology offer tools to compute harmonic centrality (e.g.,~\cite{Zhang2016}), and this centrality measure is used in analysing simulations of growth of biological networks~\cite{Paul2020}. We remark however that several papers use the term ``harmonic centrality'' to refer to a rather different centrality measure, e.g.~\cite{Ren2015,Mao2020}.

Another modification to closeness centrality we will briefly consider in this subsection was introduced in~\cite{Dangalchev2006}. This work goes beyond harmonic centrality and adds an exponential to the denominator:

\[
  \mathcal{D}(v) = \sum_{u \in V \setminus \{v\}} \frac{1}{2^{d(u,v)}}.
\]

Like harmonic centrality, $\mathcal{D}$ treats disconnected networks naturally. In addition, $\mathcal{D}$ interacts conveniently with various operations on graphs, in particular with different kinds of graph union~\cite[Section~2]{Dangalchev2006}. Finally, this centrality measure can be generalized to the following form~\cite{Dangalchev2011}:

\[
  \mathcal{D}'(v) = \sum_{u \in V \setminus \{v\}} \alpha^{d(u,v)},
\]

where $\alpha\in (0,1)$. Clearly, for $\alpha = \frac{1}{2}$, $\mathcal{D}'(v) = D_C(v)$, and as $\alpha$ increases between 0 and 1, $\mathcal{D}'$ moves from local (mostly immediate neighbours count) to global (even long-distance connections count).

Even though the centrality measures we surveyed so far in this subsection are all based on the notion of closeness, one should underline that they are not true extensions of the closeness centrality. Indeed, the paper~\cite{Yang2011} shows that, even on a 7-node tree, these centrality measures yield close, but different values.

The last centrality measure based on the length of paths to a given node that we discuss is the eccentricity centrality. It formalizes the intuition that important nodes are those from which any other node is quickly reachable. It was originally introduced in~\cite{Hage1995}, but we give here the mathematical definition from~\cite{Junker2008}.

\begin{definition}[Eccentricity centrality]

    Let $G = (V, E)$ be a strongly connected (either directed, or undirected) network. The \emph{eccentricity centrality} is the function $\mathcal{C}_e : V \to \mathbb R$ defined as follows:

    \[
        \mathcal{C}_e(v) = \frac{1}{\max_{u\in V} d(v, u)}.
    \]

\end{definition}

The value in the denominator of $\mathcal{C}_e$ is the longest \emph{shortest path} in $G$ starting at $v$ and is usually referred to as the \emph{eccentricity} of $v$.

The same definition of eccentricity centrality can be used for undirected and directed networks. In the latter case, the order of $v$ and $u$ in $d$ becomes important, and the directed network is often required to be strongly connected~\cite{Junker2008}. In networks which are not strongly connected, the shortest is computed only to nodes reachable from $v$. Moreover, if the out-degree of $v$ is 0, then by definition $\mathcal{C}_e(v)=0$.

The bounds on eccentricity centrality are the same as for normalized closeness centrality: $0 \leq \mathcal{C}_e(v) \leq 1$. $\mathcal{C}_e$ reaches its maximal value for every node which is directly connected to other nodes, as is the case of the center of a star-topology network or any node of a complete network.

Unlike closeness or harmonic centralities, eccentricity centrality does not directly depend on the size of the network, which is why normalization is not generally considered for this measure.

\begin{example}

    The following table gives the eccentricity centralities for the nodes of the example network in Figure~\ref{fig:rex}. Since the example network is not strongly connected, the centralities $\mathcal{C}_e(v)$ are computed as inverses of the lengths of the longest shortest paths to reachable nodes. The line labeled with $\mathcal{C}_e^u(v)$ takes this network to be undirected, meaning that the directed edges appearing in the figure can be traversed both ways.

    \begin{center}
        \begin{tabular}{r|ccccccccccccccc}
            $v$ & 0 & 1 & 2 & 3 & 4 & 5 & 6 & 7 & 8 & 9 & 10 & 11 & 12 & 13 & 14 \\
            \hline
            $\mathcal{C}_e(v)$ & $\frac{1}{2}$ & $\frac{1}{2}$ & $\frac{1}{3}$ & $\frac{1}{4}$ & $\frac{1}{4}$ & $\frac{1}{3}$ & $0$ & $\frac{1}{4}$ & $\frac{1}{4}$ & $\frac{1}{5}$ & $\frac{1}{3}$ & $\frac{1}{3}$ & $\frac{1}{1}$ & $\frac{1}{3}$ & $\frac{1}{4}$ \\
            $\mathcal{C}_e^u(v)$ & $\frac{1}{3}$ & $\frac{1}{3}$ & $\frac{1}{3}$ & $\frac{1}{4}$ & $\frac{1}{4}$ & $\frac{1}{4}$ & $\frac{1}{4}$ & $\frac{1}{4}$ & $\frac{1}{4}$ & $\frac{1}{5}$ & $\frac{1}{4}$ & $\frac{1}{4}$ & $\frac{1}{5}$ & $\frac{1}{4}$ & $\frac{1}{5}$ \\
        \end{tabular}
    \end{center}

\end{example}

Similarly to closeness centrality, computing the eccentricity centrality of a node $v$ requires finding the shortest paths to all other nodes in the network, meaning that computing $\mathcal{C}_e(v)$ is of time complexity $O(|V| \cdot |E|)$.

Like closeness centrality, eccentricity centrality is able to capture well the notion of importance of a node as a function of its connections to the other nodes. For example, \cite{Wuchty2003} uses several centrality measures to analyse the networks of \emph{E.~coli} and \emph{S.~cerevisiae}, and shows that both closeness and eccentricity centralities produce very similar rankings of the top metabolites. On the other hand, eccentricity centrality was not able to distinguish essential from non-essential proteins in the PPI network of \emph{S.~cerevisiae}.

\subsubsection{Path centrality: betweenness}

Betweenness centrality measures the importance of a node by counting in how many connections between other nodes it is implicated. For example, the central node 1 of $U_k^\star$ is involved in all shortest paths between all other nodes. The idea of betweenness, i.e. being situated between other nodes, was introduced in the discussion of point centrality in~\cite{Bavelas1948}, and the first formal definition was given in~\cite{Freeman1977}.

Given a (either directed, or undirected) network, we denote the set of all shortest paths (also referred to as \emph{geodesics}) between nodes $u$ and $w$ by $\rho_{uw}$, and by $\rho_{uw}(v)$ the set of those shortest paths from $\rho_{uw}$ which pass through $v$. We further denote $g_{uw} = |\rho_{uw}|$ and $g_{uw}(v) = |\rho_{uw}(v)|$. Finally, we use the following notation:

\[
    p_{uw}(v) =
    \begin{cases}
        \displaystyle \frac{g_{uw}(v)}{g_{uw}}, & \rho_{uw} \neq \emptyset, \\
        0, & \text{otherwise}.
    \end{cases}
\]

$p_{uw}(v)$ can be seen as the probability of finding $v$ in a shortest path between $u$ and $w$ chosen at random from $g_{uw}$, in the case in which $w$ is reachable from~$u$~\cite{Freeman1977,Freeman1978}.

\begin{definition}[Betweenness centrality]

    Let $G = (V, E)$ be a (either directed, or undirected) network. The \emph{betweenness centrality} is the function $\mathcal{C}_B : V \to \mathbb R$ defined as follows:

    \[
        \mathcal{C}_B(v) = \sum_{\substack{u, w\in V\setminus\{v\} \\ u \neq w}} p_{uw}(v).
    \]

\end{definition}

Note that in this case too, there are subtle distinctions between directed and undirected networks. Thus, while in the case of undirected networks $\rho_{uw}$ (as well as $\rho_{uw}(v)$) is conceptually the same as $\rho_{wu}$ ($\rho_{wu}(v)$, resp.), and consequently the pair is not considered separately, this is not the case of directed networks. Hence, while the betweenness centrality function involves a $(k-1)(k-2)/2$ summation for undirected networks, i.e., the number of un-ordered pairs of nodes distinct from $v$,  in the case of directed networks it consists of a $(k-1)(k-2)$ summation.

Betweenness centrality reaches its minimal value 0 for isolated nodes and its maximal value for nodes $v$ situated on \emph{all} shortest paths between all other nodes of the network. For a $k$-node network with $k \geq 2$, this maximal value equals the number of pairs of nodes different from $v$: $\mathcal{C}_B(v) = (k-1)(k-2)/2$ for undirected networks, and $\mathcal{C}_B(v) = (k-1)(k-2)$ for the directed ones. In the case of undirected networks, this value will be reached for the central node of a $k$-node star-topology network $U_k^\star$. In the case of directed networks, this maximal value will be reached for the central node of the $k$-node star topology network, in which there are two symmetric edges between the central node and the non-central nodes: $\bar G_k^\star = (V, E)$, with $V = \{1, \dots, k\}$ and $E = \{(1, i), (i, 1) \mid 2 \leq i \leq k\}$.

In fact, not only $\mathcal{C}_B$ reaches its maximal value for the central node of a star-topology network, but the existence of a node for which $C_B$ reaches its maximal value is sufficient to guarantee that the network is a star.

\begin{lemma}\label{lem:max-betweenness-undir}

    Let $k \geq 2$ and let $U = (V, E)$ be a $k$-node undirected network ($|V| = k$) having a node $v \in V$ for which $C_B(v) = (k-1)(k-2)/2$. Then $U$ is isomorphic to $U_k^\star$.

\end{lemma}

\begin{proof}

    The fact that $C_B(v) = (k-1)(k-2)/2$ means that $v$ appears in \emph{all} shortest paths between all other nodes. This implies on the one hand that $v$ is connected to all other nodes, and on the other hand that there are no other edges in $U$. Indeed, if $v$ were not connected to some of the other nodes, then there would exist a pair of nodes $u, w \in V$ with the property $g_{uw}(v) = 0$, meaning that $C_B(v) < (k-1)(k-2)/2$. On the other hand, if there existed an edge between two vertices $u, w \in V$, then the only shortest path from $u$ to $w$ would not traverse~$v$, which again would mean that $g_{uw}(v) = 0$ and $C_B(v) < (k-1)(k-2)/2$.

\end{proof}

A similar result for directed networks imposes that any directed network $G$ having a node with maximal betweenness centrality should be isomorphic to $\bar G_k^\star$.

\begin{lemma}\label{lem:max-betweenness-dir}

    Let $k \geq 2$ and let $G = (V, E)$ be a $k$-node directed network ($|V| = k$) having a node $v \in V$ for which $C_B(v) = (k-1)(k-2)$. Then $G$ is isomorphic to $\bar G_k^\star$.

\end{lemma}

Note that in the case of the star-topology directed network $G_k^\star$ in which there are arcs going from all non-central nodes to the central node $1$ and no arcs going out of $1$, the betweenness centrality of all nodes will be 0, because all paths in this network have length 1.

To avoid the increase in betweenness centrality due only to the increase in the size of the network, one typically defines its normalized version. Unlike normalized closeness and harmonic centralities, in the case of betweenness one needs to normalize with respect to the number of pairs of nodes, rather than the number of nodes.

\begin{definition}[Normalized betweenness centrality]

    Let $G = (V, E)$ be a (either directed, or undirected) network. The \emph{normalized betweenness centrality} is the function $\widetilde{\mathcal{C}}_B : V \to \mathbb R$ defined as follows:

    \begin{align*}
        \mbox{a)} & \widetilde{\mathcal{C}}_B(v) = \frac{2}{(|V|-1)(|V|-2)} \sum_{\substack{u, w\in V\setminus\{v\} \\ u \neq w}} p_{uw}(v), & \mbox{ for undirected networks}, \\
        \mbox{b)} & \widetilde{\mathcal{C}}_B(v) = \frac{1}{(|V|-1)(|V|-2)} \sum_{\substack{u, w\in V\setminus\{v\} \\ u \neq w}} p_{uw}(v), & \mbox{ for directed networks}.
    \end{align*}

\end{definition}

Normalized betweenness centrality ranges between 0 for isolated vertices and 1 for central notes of star-topology networks.

\begin{example}

    The following table gives the betweenness centralities and the normalized betweenness centralities, rounded to two digits after the decimal point, for the nodes of the example network in Figure~\ref{fig:rex}. The lines $\mathcal{C}_B^u(v)$ and $\widetilde{\mathcal{C}}_B^u(v)$ take this network to be undirected, meaning that the directed edges appearing in the figure can be traversed both ways.

    \begin{center}
        \begin{tabular}{r|ccccccccc}
          $v$ & 0 & 1 & 2 & 3 & 4 & 5 & 6 & 7 \\
          \hline
          $\mathcal{C}_B(v)$ & 12.00 & 23.00 & 21.00 & 0.00 & 4.00 & 0.00 & 0.00 & 0.00 \\
          $\widetilde{\mathcal{C}}_B(v)$ & 0.07 & 0.13 & 0.12 & 0.00 & 0.02 & 0.00 & 0.00 & 0.00 \\
          \hline
          $\mathcal{C}_B^u(v)$ & 0.00 & 68.00 & 54.00 & 0.00 & 13.00 & 0.00 & 13.00 & 0.00 \\
          $\widetilde{\mathcal{C}}_B^u(v)$ & 0.00 & 0.75 & 0.59 & 0.00 & 0.14 & 0.00 & 0.14 & 0.00 \\
        \end{tabular}
    \end{center}

    \begin{center}
        \begin{tabular}{r|cccccccc}
          $v$ & 8 & 9 & 10 & 11 & 12 & 13 & 14 \\
          \hline
          $\mathcal{C}_B(v)$ & 0.00 & 0.00 & 0.00 & 0.00 & 0.00 & 4.00 & 0.00 \\
          $\widetilde{\mathcal{C}}_B(v)$ & 0.00 & 0.00 & 0.00 & 0.00 & 0.00 & 0.02 & 0.00 \\
          \hline
          $\mathcal{C}_B^u(v)$ & 0.00 & 0.00 & 0.00 & 0.00 & 0.00 & 13.00 & 0.00 \\
          $\widetilde{\mathcal{C}}_B^u(v)$ & 0.00 & 0.00 & 0.00 & 0.00 & 0.00 & 0.14 & 0.00 \\
        \end{tabular}
    \end{center}

\end{example}

Enumerating all shortest paths between all pairs of nodes different from a given node $v$ is rather expensive. The seminal work~\cite{Freeman1977} suggests computing powers of the adjacency matrix of the network to count the shortest paths, according to the methods detailed in~\cite{Harary1965}. However, \cite{Brandes2001} proposes a less expensive algorithm running in time $O(|V| \cdot |E|)$, and avoiding extra non-optimal paths which matrix multiplication would reveal. A more recent work~\cite{Nasre2014} proposes an even faster incremental algorithm, which updates the betweenness centralities of all nodes when a new edge is added. This algorithm runs in time $O(m |V|)$, where $m$ is the maximal length of the shortest path in the network. Finally, the paper~\cite{Borassi2019} proposes a random approximation algorithm which computes betweenness centralities in time $|E|^{1/2 + o(1)}$ with high probability, much faster than exact algorithms.

Betweenness centrality has been successfully applied in biological network analysis to identify significant nodes. For example, the work~\cite{Sahoo2016} used betweenness centrality to identify proteins significant in the context of cancer growth. In~\cite{Yu2007}, the authors  rely on betweenness centrality to identify bottlenecks in PPI networks. Moreover, the authors show that these nodes have a high tendency of representing  essential proteins. The work~\cite{Joy2005} shows that nodes with high betweenness centrality in yeast PPI networks are more likely to be essential and that the evolutionary age of proteins is positively correlated with their betweenness centrality.

We conclude this subsection by mentioning a well-known extension of the betweenness centrality: \emph{percolation centrality}. This measure was introduced in~\cite{Piraveenan2013} with the goal of taking into account the dynamical states of the nodes of the network. Intuitively, percolation centrality extends the definition of betweenness centrality by multiplying $p_{uw}(v)$ by \emph{percolation}: a factor derived from the degree to which the nodes in the paths from $u$ to $w$ are involved in the spread of some value (information, infection, etc.) through the network. Since percolation centrality considers network states, we do not discuss it in this section.

\subsubsection{Spectral centralities}
\label{sec:spectral-centralities}

In this section we discuss eigenvector-based centralities: the eigenvector centrality, Katz centrality, and PageRank centrality.

The idea of using eigenvectors for assessing the importance of nodes in networks was first introduced in~\cite{Bonacich1987} and further discussed in several fundamental works, e.g.~\cite{Ruhnau2000,Newman2008} with the goal of taking into account the fact that not all connections in a network are equal --- being connected to a highly central node has more impact than being connected to a more peripheral node. Before defining eigenvector centrality, we introduce some additional notations.

Given a (either directed, or undirected) network $G = (V, E)$, we denote its adjacency matrix by $\mathbf A$, in which $A_{uv} = 1$ if and only if $G$ contains an edge from node $u$ to node $v$. For a centrality measure $\mathcal{C} : V \to \mathbb R$, we will denote by $\mathbf c$ the vector obtained by computing $\mathcal{C}$ for every node of $G$, i.e. ${\mathbf c}_v = \mathcal{C}(v)$.

For a fixed node $v$, the idea that the contributions of its neighbors are proportional to their own centralities can be expressed as follows:

\[
    \mathcal{C}(v) = \frac{1}{\lambda} \sum_{u \in V} A_{uv} \mathcal{C}(u),
\]

where $\lambda > 0$ is a real constant. This relation can be rewritten in the matrix form:

\begin{equation}
    \label{eq:eigenvector}
    \lambda \mathbf{c} = \mathbf{A} \mathbf{c},
\end{equation}

which effectively defines $\mathbf c$ as an eigenvector, and $\lambda$ as an eigenvalue of the adjacency matrix $\mathbf A$ of the network $G$. Since $\mathbf A$ is a real square matrix (its elements are 0 and 1 in the case of unweighted networks), according to the Perron-Frobenius theorem of linear algebra, $\mathbf A$ has a unique largest real eigenvalue $\lambda$ and one can choose the corresponding eigenvector $\mathbf c$ to be strictly positive, which allows defining a meaningful centrality measure~\cite{Newman2008}.

\begin{definition}[Eigenvector centrality]
    \label{def:eigenvector-centrality}

    Let $G = (V, E)$ be a (either directed, or undirected) network and $\mathbf A$ its adjacency matrix. Let $\lambda$ be the largest real eigenvalue of $\mathbf A$ and $\mathbf c$ a corresponding eigenvector with strictly positive components. An \emph{eigenvector centrality} is a function $\mathcal{C}_E : V \to \mathbb R$ defined as follows: $\mathcal{C}_E(v) = \mathbf c_v$.

\end{definition}

Eigenvalue centrality is sometimes referred to as eigencentrality or eigenvector prestige. The same definition can be used for directed networks, in which case the adjacency matrix $\mathbf A$ may not be diagonally symmetric.

We will refer to eigenvectors $\mathbf c$ satisfying the constraints of the previous definition as \emph{principal} eigenvectors of $\mathbf A$ and of the network $G$. In fact, it follows from~(\ref{eq:eigenvector}) that any other vector $\alpha \mathbf c$, $\alpha > 0$, is also a principal eigenvector of $\mathbf A$. This means that the components of $\mathbf c$ (and therefore the values of $\mathcal{C}_E$) indicate relative centralities of the nodes of $G$ with respect to one another, rather than some absolute centrality score.

A standard way of fixing a preferred principal eigenvector $\mathbf c$ to avoid ambiguity is by normalizing $\mathbf c$. Several ways to normalize eigenvectors for centrality exist; we focus here on three norms studied in~\cite{Ruhnau2000} and we discuss implications of these normalizations for the eigenvector centrality index --- a network-wide measure of centralization. The results from~\cite{Ruhnau2000} also rely on the work of~\cite{Papendieck2000} which studies maximal entries in the principal eigenvector of a graph.

Normalizing an $n$-vector $\mathbf a$ typically consists in dividing its components by a $p$-norm, which is commonly defined as follows:

\[
    \| \mathbf a \|_p =
    \begin{cases}
        \displaystyle \left( \sum_{i=1}^n a_i^p \right)^\frac{1}{p}, & 1 \leq p < \infty, \\
        \max_{1\leq i \leq n}| a_i |, & p = \infty.
    \end{cases}
\]

One of the ways of normalizing the vector $\mathbf c$ from Definition~\ref{def:eigenvector-centrality} is by dividing all of its components by the $\infty$-norm, also known as the maximum norm.

\begin{definition}[$\infty$-norm eigenvector centrality]

    Let $G = (V, E)$ be a (either directed, or undirected) network and $\mathbf c$ be the principal eigenvector of $G$. The \emph{$\infty$-norm eigenvector centrality} is the function $\mathcal{C}_E^{(\infty)} : V \to \mathbb R$ defined as follows:

    \[
        \mathcal{C}_E^{(\infty)}(v) = \frac{\mathbf c_v}{\| \mathbf c \|_\infty} = \frac{\mathbf c_v}{\max_{u \in |V|} \mathbf c_u}.
    \]

\end{definition}

It follows directly from this definition that $0 \leq \mathcal{C}_E^{(\infty)}(v) \leq 1$, and that every connected network $G$ has a node $v^*$ for which $\mathcal{C}_E^{(\infty)}(v^*) = 1$.

Another way of normalizing the eigenvector centrality is by using the 1-norm, also known as the sum norm.

\begin{definition}[1-norm eigenvector centrality]

    Let $G = (V, E)$ be a (either directed, or undirected) network and $\mathbf c$ be the principal eigenvector of $G$. The \emph{1-norm eigenvector centrality} is the function $\mathcal{C}_E^{(1)} : V \to \mathbb R$ defined as follows:

    \[
        \mathcal{C}_E^{(1)}(v) = \frac{\mathbf c_v}{\| \mathbf c \|_1} = \frac{\mathbf c_v}{\sum_{u \in V} \mathbf c_u}.
    \]

\end{definition}

$\mathcal{C}_E^{(1)}(v)$ can be seen as the proportion of centrality that $v$ reaches in $G$. The bounds on the values of the 1-norm eigenvector centrality are as follows~\cite{Ruhnau2000}:

\[
    0 \leq \mathcal{C}_E^{(1)}(v) \leq \frac{1}{1 + (2 \cos \frac{\pi}{n+1})^{-1}}.
\]

This centrality measure cannot ever reach 1 for any network with 2 or more connected nodes. The upper bound is only reached for the nodes of a network which only contains one edge~\cite{Papendieck2000}. In the case of networks with more than 2 connected nodes, it is not known what is the actual maximal value $\mathcal{C}_E^{(1)}$ can achieve, nor in which kind of network topologies this value can be achieved. For example, the central node of a star-topology network \emph{does not} reach the maximal value.

One last way of normalizing the eigenvector centrality which we consider in this section is by using the 2-norm, also known as the Euclidean norm.

\begin{definition}[2-norm eigenvector centrality]

    Let $G = (V, E)$ be a (either directed, or undirected) network and $\mathbf c$ be the principal eigenvector of $G$. The \emph{2-norm eigenvector centrality} is the function $\mathcal{C}_E^{(2)} : V \to \mathbb R$ defined as follows:

    \[
        \mathcal{C}_E^{(2)}(v) = \frac{\mathbf c_v}{\| \mathbf c \|_2} = \frac{\mathbf c_v}{\sqrt{\sum_{u \in V}^{\phantom{a}} \mathbf c_u^2}}.
    \]

\end{definition}

The bounds on the 2-norm eigenvector centrality are $0 \leq \mathcal{C}_E^{(2)}(v) \leq \frac{1}{\sqrt{2}}$. The maximal value $\frac{1}{\sqrt{2}}$ is only reached by the center of the star topology network~\cite{Ruhnau2000}. One may define a derived measure ranging from 0 to~1 by multiplying $\mathcal{C}_E^{(2)}$ by $\sqrt{2}$.

Finding eigenvectors and eigenvalues of a given square matrix is closely related to finding the roots of polynomials. It follows from the Abel-Ruffini theorem that there exists  no algorithm  computing exactly the eigenvectors and eigenvalues for square matrices of size greater than 4 (e.g.,~\cite{Golub2000}). Therefore, iterative algorithms are usually used, one of the most popular being the power iteration method (used e.g., in \emph{NetworkX}~\cite{networkx-article}). The complexity of such iterative methods is generally between $O(|V|^2)$ and $O(|V|^3)$, while the convergence rate ranges from linear to cubic~\cite{Demmel1997}.

Eigenvalue centrality measures are a fine instrument for measuring the importance of nodes in a network, because they acknowledge the difference in impact between a connection to a high-centrality neighbor and a connection to a low-centrality one. With eigenvalue centrality, a node with a smaller number of ``high-quality'' connections may outrank a node with a larger number of ``low-quality'' connections~\cite{Newman2008}. The work~\cite{Estrada2006} applies closeness, betweenness, and eigenvector centralities to identifying essential proteins in PPI networks, and concludes that spectral centralities show the best performance. The paper~\cite{Negre2018} uses eigenvector centrality to pinpoint key amino acids in terms of their relevance in the allosteric regulation. The study~\cite{Melak2015} uses different centrality measures to identify potential drug targets of Mycobacterium tuberculosis, the etiological agent of tuberculosis (TB), and show that eigenvalue centrality fares best.

Another centrality measure based on the algebraic properties of the adjacency matrix is \emph{Katz centrality} (also referred to as \emph{Katz prestige} or \emph{Katz status index}). This method was proposed in~\cite{Katz1953} and it gives a centrality score by taking into consideration all the nodes of a given network. According to a reasoning similar to the one made for closeness centrality and its variants (Subsection~\ref{sec:proximity-centralities}), a node is of high importance if it is connected to many other nodes, but nodes situated farther away count less toward the total centrality score.

\begin{definition}[Katz centrality]

    Let $G = (V, E)$ be a (either  directed, or undirected)  loop-free network. Let $\mathbf A$ be its adjacency matrix and $\lambda$ be its largest positive eigenvalue. The \emph{Katz centrality} is the function $\mathcal{C}_K : V \to \mathbb R$ defined as follows:

    \[
        \mathcal{C}_K(v) = \sum_{i=1}^\infty \sum_{u\in V} \alpha^i (\mathbf A^i)_{uv},
    \]

    where $\alpha$ is a constant, $0 \leq \alpha \leq 1/\lambda$, and $(\mathbf A^i)_{uv}$ is the element in row $u$ and column $v$ of the $i$-th power of $\mathbf A$.

\end{definition}

Since $(\mathbf A^i)_{uv}$ is non-zero if and only if there exists a path of length exactly $i$ between $u$ and $v$, Katz centrality can be interpreted as a generalization of degree centrality~\cite{Junker2008}. Indeed, the role of the factor $\alpha^i$ is to scale down the contributions of longer paths, and without it $\mathcal{C}_K(v)$ essentially becomes the \emph{reachability index}: the number of nodes from which $v$ can be reached. If $\alpha$ is close to $0$, the contributions of longer paths are essentially discarded, and $\mathcal{C}_K$ approaches a form of degree centrality. On the other hand, as $\alpha$ approaches $1/\lambda$, $\mathcal{C}_K$ approaches eigenvector centrality~\cite{Newman2010}.

Let $\mathbf c_K$ be the vector collecting Katz centralities of a node in a given network: $(\mathbf c_K)_v = \mathcal{C}_K(v)$. Then $\mathbf c_K$ can be expressed in a more compact matrix multiplication form in the following way~\cite{Junker2008}:

\[
    \mathbf c_K = ((\mathbf I - \alpha A^T)^{-1} - \mathbf I) \mathbf 1,
\]

where $\mathbf I$ is the identity matrix of the same size as $\mathbf A$, $\mathbf A^T$ is the transpose of $\mathbf A$, $(\cdot)^{-1}$ denotes matrix inversion, and $\mathbf 1$ is a $|V|$-vector whose components are all~1.

Due to similarities with eigenvector centrality, exact Katz centrality can be computed by similar algorithms with similar running time complexities. In particular, the power iteration method can be applied (e.g.,~\cite{networkx-katz}). Faster approximate algorithms exist, in particular~\cite{Foster2001} presents an iterative algorithm with time complexity $O(|V| + |E|)$, given a constant desired precision.

Typical applications of Katz centrality concern directed networks, in particular directed acyclic networks, in which eigenvector centrality appears less useful~\cite{Newman2010}. Katz centrality has found promising applications in neuroscience. The work~\cite{Fletcher2018} shows that Katz centrality is the best predictor of firing rate given the network structure, with almost perfect correlation in all cases studied. The paper~\cite{Mantzaris2013} uses Katz centrality to analyse fMRI data of human brain activity during learning, and shows how key brain regions contributing to the process can be discovered.

A third eigenvalue and eigenvector-related centrality measure we briefly discuss in this section is PageRank, one of the algorithms used by Google to measure the importance of web pages~\cite{Google2011}. Multiple variants of PageRank have been proposed. We start with the definition from~\cite{Langville2005}.

\begin{definition}[PageRank]

    Let $G = (V, E)$ be a directed network and $\mathbf A$ its adjacency matrix. The \emph{PageRank} centrality is the function $\mathcal{C}_P : V \to \mathbb R$ satisfying the following equation:

    \[
        \mathcal{C}_P(v) = \sum_{u \in V} A_{uv} \frac{ \mathcal{C}_P(u)}{deg^+(u)},
    \]

    where $deg^+(u)$ is the out-degree of the node $u$.

\end{definition}

According to this definition, the centrality $\mathcal{C}_P(v)$ depends on the centralities of the nodes $u$ from which there is an edge going to $v$ ($A_{uv} \neq 0$). The contribution of each of these nodes $u$ to $\mathcal{C}_P(v)$ is inversely proportional to the number of edges going out of $u$. Thus, nodes connected to many other nodes contribute less to each of their neighbors than nodes which are connected to fewer nodes.

The complete definition of PageRank centrality (which was used in the first versions of the Google search engine) includes two additional parameters:

\[
    \mathcal{C}_P^{(0)}(v) = \alpha \sum_{u \in V} A_{uv} \frac{ \mathcal{C}_P(u)}{deg^+(u)} + \beta,
\]

where $\alpha$ can be seen as a decay factor, and $\beta$ as a vector of initial centralities (source ranks), scaled by $\alpha$. We take this definition from~\cite{FranceschetPageRank}, which gives a slightly simplified and generalized version of the original definition from~\cite{Page1999}.

The PageRank centrality is related to eigenvector centrality in the way the adjacency matrix of the network is exploited to define relative centrality scores. Since the formulae for computing PageRank are more complex, using iterative algorithms is often preferred, especially for large networks~\cite{Langville2005,Page1999}. Several specific and faster algorithms have been proposed for computing PageRank. In particular, the paper~\cite{DelCorso2005} proposes an algorithm consisting in reducing PageRank computation to computing solutions of linear systems, while the article~\cite{Bahmani2010} proposes Monte Carlo methods for incremental computation of PageRank.

PageRank is a relatively recent centrality measure, but it has already shown some promising performance in analysis of biological networks. In the work~\cite{Gabor2010}, the authors computed PageRank centralities for the metabolic network of the Mycobacterium tuberculosis and the PPI networks of melanoma patients, and in both cases important proteins received a high centrality score.

\begin{example}

    The following table gives eigenvector ($\mathcal{C}_E^X$), Katz ($\mathcal{C}_K^X$), and PageRank ($\mathcal{C}_P^X$) centralities for the nodes of the example network in Figure~\ref{fig:rex}, as computed with \emph{NetworkX}~\cite{networkx-url} and rounded to two digits after the decimal point.

    \begin{center}
        \begin{tabular}{r|ccccccccc}
            $v$ & 0 & 1 & 2 & 3 & 4 & 5 & 6 & 7 \\
            \hline
            $\mathcal{C}_E^X$ & 0.50 & 0.50 & 0.50 & 0.00 & 0.00 & 0.00 & 0.50 & 0.00 \\
            $\mathcal{C}_K^X$ & 0.26 & 0.35 & 0.36 & 0.23 & 0.25 & 0.23 & 0.29 & 0.23 \\
            $\mathcal{C}_P^X$ & 0.19 & 0.25 & 0.20 & 0.02 & 0.03 & 0.02 & 0.14 & 0.02 \\
        \end{tabular}
    \end{center}

    \begin{center}
        \begin{tabular}{r|cccccccc}
            $v$ & 8 & 9 & 10 & 11 & 12 & 13 & 14 \\
            \hline
            $\mathcal{C}_E^X$ & 0.00 & 0.00 & 0.00 & 0.00 & 0.00 & 0.00 & 0.00 \\
            $\mathcal{C}_K^X$ & 0.23 & 0.23 & 0.23 & 0.23 & 0.23 & 0.25 & 0.23 \\
            $\mathcal{C}_P^X$ & 0.02 & 0.02 & 0.02 & 0.02 & 0.02 & 0.03 & 0.02 \\
        \end{tabular}
    \end{center}

\end{example}

%% file: subsection-methods-controllability.tex
\subsection{System controllability methods}
\label{subsection-methods-controllability}

One of the important aspects when dealing with a biomedical network is to be able to influence part of, or even the totality of its nodes. This objective is not restricted to biomedical networks, but is one of the earliest and most studied network theory problems, generally known as the (target) network controllability problem. This is intrinsically an optimization problem, as any network can be controlled from a sufficiently large controlling set, e.g., the entire set of nodes. However, one is interested in finding the minimal set of nodes needed in order to achieve such control. Early works in the field come from the 60's, however it was in early 2010's that one of the key results from the field has been provided in~\cite{Liu2011}, namely the possibility of efficiently determining the minimum set of nodes needed to control an entire network. This result sparked a new interest in the field, with a strong emphasis on possible applications in the biomedical field.

The prospect of enforcing control over a biomedical network has been investigated within two main methodological approaches: that of (structural) network controllability, see, e.g.,~\cite{Liu2011,Czeizler18,Guo18,Kanhaiya17} and that of dominating sets~\cite{Wuchty14,Nacher16,Zhang2015}. In the current section we review both approaches, detailing the theoretical, algorithmic and biomedical applicability aspects of these methods.

\subsubsection{Network Controlability}

In general terms, we say that a directed network, or generally any dynamical system, is controllable from a set of input nodes if, with a suitable selection of input values for these nodes, the entire network can be driven from any initial state to any desired final state within a finite time. The state of a node is given by its numerical value; in the context of biomedical networks this could be for example the expression level of a gene/protein. From one time point to another this value/state evolves depending (linearly) on the values/states of its neighbors; this is why such systems are also known as linear time invariant dynamical systems (LTIS), as their evolution 
can be described by the system of linear differential equations:

\begin{equation}
    \label{no_in_sys}
    \frac{dx(t)}{dt}=Ax(t),
\end{equation}
where $x(t)=(x_1(t), ..., x_n(t))^T$ is the $n$-dimensional vector describing the system's state at time $t$, and $A\in \mathbb R^{n\times n}$ is the time-invariant \textit{state transition matrix}, describing how each of these states are influencing the dynamics of the system. Namely, for any nodes $i,j$, $1\leq i,j\leq n$ within the network, the entry $a_{i,j}$ of matrix $A$ either documents the weight of the influence of node $j$ over the node $i$, if there exists a (directed) edge $(i,j)$, or is equal to 0 otherwise. By convention, from now on during this section, all vectors are considered to be column vectors so that the matrix-vector multiplications are well defined.

Assume we allow the system to be influenced through an $m$-dimensional input controller, i.e. an input vector $u$ of real functions, $u:\mathbb{R}\rightarrow\mathbb{R}^m$ acting upon some $m$ nodes of the network. Consider the subset of input nodes $I\subseteq \{1,2,\ldots, n\}$, $I=\{i_1,\ldots,i_m\}$, $1\leq m\leq n$, as the nodes of the network on which the external input is applied to; such nodes are also known as \emph{driver/driven} nodes. The system~(\ref{no_in_sys}) becomes:

\begin{equation}
    \frac{dx(t)}{dt}=Ax(t) + B_Iu(t), \label{in_sys}
\end{equation}

where $B_I\in\mathbb{R}^{n\times m}$ is the characteristic matrix associated to the subset $I$, $B_I(r,s)=1$ if $r=i_s$ and $B_I(r,s)=0$ otherwise, for all $1\leq r\leq n$ and $1\leq s\leq m$.

It is often the case, particularly in the bio-medical field, that it is enough to enforce control only over a subset of the network nodes to get the desired change in the network dynamics. The associated optimization problem is known as the \emph{target controllability problem}, or more generally as the \emph{output controlability problem}. Thus, consider a subset of target nodes $T\subseteq\{1,2,\ldots,n\}$, $T=\{t_1,\ldots,t_l\}$, $m\leq l\leq n$, thought of as a subset of the nodes of the linear dynamical system, whose dynamics we aim to control (as defined below) through a suitable choice of input nodes and of an input vector. The subset of target nodes can also be defined through its characteristic matrix $C_T\in \mathbb{R}^{l\times n}$, defined as $C_T(r,s)=1$ if $t_r=s$ and $C_T(r,s)=0$ otherwise, for all $1\leq r\leq l$ and $1\leq s\leq n$. We will denote by the triplet $(A,I,T)$ the targeted linear (time-invariant) dynamical system defined by matrix $A$, with input set $I$, and target set $T$.  In case $T$ is the entire network, we can omit it from the notation.

\begin{definition}[Target controllability of linear networks]

    Given a target linear dynamical system $(A,I,T)$ (or similarly a linear network), we say that this system is \emph{target controllable} (or simply controllable if the target is the entire network) if for any $x(0)\in\mathbb{R}^n$ and any $\alpha\in\mathbb{R}^l$, there is an input vector $u:\mathbb{R}\rightarrow\mathbb{R}^m$ such that the solution $\tilde{x}$ of \eqref{in_sys} eventually coincides with $\alpha$ on its $T$-components, i.e., $C_T\tilde{x}(\tau)=\alpha$, for some $\tau\geq 0$.

\end{definition}

Note that, the general case of \emph{output controlability} is defined similarly, with the only difference that instead of a 0--1 characteristic matrix $C_T$ we have an arbitrary matrix $C\in\mathbb{R}^{n\times m}$, with $m\leq n$, defining the output $y\in\mathbb{R}^m$ of the LTIS as a linear combination of the solution $x(t)$ of equation~(\ref{in_sys}).

Intuitively, the system being target controllable means that for any input state $x_0$ and any desired final state $\alpha$ of the target nodes, there is a suitable input vector $u$ driving the target nodes to $\alpha$. Obviously, the input vector $u$ depends on $x_0$ and $\alpha$.

It is known from~\cite{Kal63} that the system $(A,I)$ is controllable from some input controller $u(t)$ if and only if the rank of the matrix $(B_I \mid AB_I \mid A^2B_I \mid\ldots\mid A^{n-1}B_I)$, known as the \emph{controlability matrix}, is equal to $n$, the number of nodes in the network.  The operator $|$ denotes here the simple matrix concatenation operation. This criterion is known as the \emph{Kalman's criterion for full controllability}~\cite{Kal63}, and it is easily extendable to the case of target (or generally, output) controllability:

\begin{theorem}
    \label{Kal}

    A targeted linear dynamical system with inputs $(A,I,T)$ is controllable if and only if the rank of its controllability matrix $(C_TB_I\mid C_TAB_I\mid C_TA^2B_I\mid\ldots\mid C_TA^{n-1}B_I')$ is equal with $|T|$~\cite{Kal63}.

\end{theorem}

Intuitively, the controllability matrix describes all weighted paths from the input nodes to the target nodes in the directed graph associated to the linear dynamical system. This leads to the notion of \emph{control path} from an input node to a target node and an input node \emph{controlling} a target node. This line of thought can be further developed into a structural formulation of the targeted controllability. Thus, although the Kalman controlability criterion seems to suggest that the (target) controllability problem is strictly related to a particular valuation of the state transition matrix $A$, it turns out that this is actually a network property. To explain this, we define two matrices $C$, $D$ of the same size to be \emph{equivalent} if they have the same set of non-zero entries. The key concept here is that of \emph{structural controlability}:

\begin{definition}[Structural target controllability of linear networks] We say that a given targeted linear dynamical system $(A,I,T)$ (or similarly a linear network) is \emph{structural target controllable} if there exists a matrix $A'$ equivalent with $A$ such that the system $(A',I,T)$ is target controllable. That is, $(A,I,T)$ becomes controllable by replacing $A$ with any suitable equivalent matrix $A'$, or similarly, by replacing the weights of the network's edges with any other non-zero values. 
\end{definition}

They key result, proved in \cite{Lin1974, Shi76}, is that if a system $(A,I)$ is structurally controllable, then it is controllable for all, except a \emph{thin} set, of equivalent matrices $A'$. We recall that a thin subset of the $n$-dimensional complex space is nowhere dense and has Lebesgue measure 0. This indicates that the controlability problem can indeed be reduced to its structural version, known as the \emph{structural (target) controllability problem}, which is ultimately a property of the network's connectivity rather than of the effective edge weights. This problem can be defined both as a decision problem (is the linear time-invariant dynamical system $(A,I,T)$ structurally target controllable?) as well as an optimization problem (given matrix $A$ of size $n\times n$ and a subset of target nodes $T\subseteq\{1,2,\ldots,n\}$, find the minimal $m$, $m\leq n$, and a suitable choice for the size-$m$ set of input nodes $I$ such that $(A,I,T)$ is structurally target controllable). We are most interested in the latter formulation of the problem. 

From an algorithmic standpoint, the Kalman criterion for controllability from Theorem~\ref{Kal}  can lead to an efficient/polynomial time algorithm for verifying whether a given network can be (target) controlled from a given input controller, acting upon some of the network nodes. However, it does not lead to an efficient way of computing the composition of such a minimal controller, or even its (minimal) size. As it turns out, the (target) control optimization problem has been open for more than 20 years, before it was shown in~\cite{Liu2011} that there exists a low polynomial time algorithm (i.e., cubic in the size of the network) that can provide both the size and the composition of the minimum input controller needed to control an entire network. It was thus surprising that the corresponding result for the target control optimization problem was proved by~\cite{Czeizler18} to be NP-hard, meaning that any exact optimization algorithm would have to run in exponential time. Several approximation algorithms have been developed for this latter problem, with good results when applied to both real-word networks, e.g. bio-medical, social, electrical etc., as well as artificial networks~\cite{Kanhaiya17,Czeizler18,Gao14}. 

Another generalization coming from the practical application of network controllability in pharmacology and biomedicine considers the case when the input controller should, or it is desired to, be selected mostly from a subset $P\subseteq \{1,2,\ldots, n\}$ of the network nodes. For example, in network pharmacology studies, it is more advantageous to consider controllers consisting of those genes/proteins for which there already exist approved drugs known to target that particular element. This leads to the so-called \emph{input-constrained targeted structural controllability} problem: 

\begin{definition}[Input-constrained, structural target controllability]
    
    Given a linear dynamical system defined by a matrix $A\in\mathbb{R}^{n\times n}$, a set of target nodes $T$ and a set of preferred nodes $P$, the \emph{input-constrained structural targeted controllability} problem asks to find a smallest-sized input set $I$ whose intersection with $P$ is maximal, such that the targeted linear dynamical system with inputs $(A,I,T)$ is controllable, i.e., such that the rank of the matrix $(C_TB_I\mid C_TAB_I\mid C_TA^2B_I\mid\ldots\mid C_TA^{n-1}B_I)$ is equal with $|T|$.

\end{definition}

While this optimization problem is also NP-hard (as a generalization of the previous case), several efficient approximation algorithms have been introduced in~\cite{Kanhaiya17} and~\cite{Popescu20} and analyzed particularly with respect to bio-medical networks.

\subsubsection{Minimum dominating sets}

The second frequently used controllability technique, applied especially in the case of undirected networks, relies on the notion of domination. 

\begin{definition}[Dominating set]
    
    Given a size $n$ network, we say that a subset $D\subseteq \{1,2,\ldots, n\}$ of its nodes is \emph{dominating} the network if any node within the network is either in $D$ or it is adjacent to a node in $D$.

\end{definition}

\begin{definition}[The Minimum Dominating Set Problem]

    Given a (undirected) network, the \emph{Minimum Dominating Set Problem (MDSP)} asks to find a dominating set of minimum cardinality. Such MDS can thus be considered as an efficient first-hand controlling set for the respective network. 

\end{definition}

From an algorithmic point of view MDSP is a classical NP-hard problem~\cite{Garey90}. However, as before, there are many efficient approximation algorithms providing efficient solutions even in the case of large networks~\cite{Hedar12,Alon1990,Nacher14}. 

The MDS approach has been applied in connection to various fields such as the controllability  of  biological  networks~\cite{Wuchty14,Nacher16,Zhang2015}, design and analysis of wireless computer networks~\cite{Yu13,Wu06}, the study of social networks~\cite{Daliri18,Wang11}, etc. Several generalizations of MDSP have also been considered:

\begin{definition}[The Minimum k-Dominating Set Problem]

    Given a (undirected) network, a set $D$ of its nodes is \emph{$k$-path dominating} if any node from the network is either in $D$ or it is connected to a node in $D$ through a path of length at most $k$. The \emph{minimum k-dominating set problem} (MkDP), also known as the \emph{d-hop dominating} problem, asks to find a $k$-path dominating set of minimal cardinality.

\end{definition}

\begin{definition}[The Red--Blue ($k$-)Dominating Problem]
    
    Given a network and two subsets of its nodes, Red and Blue, the \emph{Red--Blue ($k$-)do\-mi\-nation problem}, asks to find a minimum subset of the Blue nodes ($k$-path) dominating all the Red nodes.

 \end{definition}
 
Such generalizations were considered in~\cite{Lucia04,Nguyen20,Santos17} and in~\cite{Faisal11}, respectively. All these generalizations, while preserving the algorithmic complexity of the original MDSP, provide sometimes a closer connection to practical applications. 

Compared to the structural network controllability formalism, the dominating set methodology enforces a reachability type of control. Indeed, in the latter case, the dominating nodes are controlling the network either by direct interactions, in case of MDSP, or by paths of length at most $k$, for MkDSP. This is different than within the structural network controllability formalism, where each controller node is asked to enforce an independent control over its dominating nodes. This requirement imposes that each controlled node, i.e. \emph{target} node, is positioned at the end of a different lengthed path starting from the controller, i.e., the \emph{driver} node.

%% file: subsection-methods-software.tex
\subsection{Software}
\label{subsection-methods-software}

In this section we briefly present several applications and libraries which can be used for generating, visualizing or analyzing graphs and networks. All of the mentioned software is open-source, free to use and still supported as of the time of writing (namely, it was updated in the past two years and it runs on the latest versions of the most common operating systems or using the most recent versions of the corresponding programming languages).

\subsubsection{NetworkX}
\label{subsubsection-methods-software-networkx}

\textit{NetworkX} (\cite{networkx-article}) is a Python package for the creation, ma\-ni\-pu\-la\-tion, and study of the structure, dynamics, and functions of complex networks (\cite{networkx-url}). As all of the networks described in this survey (except the simple demonstrative network from Figure ~\ref{fig:rex}) represent personalized protein-protein interaction networks and have been generated externally, we have focused the usage of the \textit{NetworkX} package solely on analyzing them. The package provides out-of-the-box implemented algorithms and functions for identifying all of the centrality measures presented in this survey. It is worth noting that \textit{NetworkX} also provides several algorithms for different layout routines, and basic interconnections with dedicated graph visualization packages, such as \textit{Matplotlib} (\cite{matplotlib-url}) and \textit{Graphviz} (\cite{graphviz-url}). However, graph visualization does not represent the main goal of \textit{NetworkX}, and its creators recommend using a dedicated and fully-featured tool instead.

An alternative to \textit{NetworkX} is the \textit{igraph} collection (\cite{igraph-url}), which provides network analysis libraries and packages for R, Mathematica and C/C++ in addition to Python.

We used the \textit{NetworkX} package to compute the centrality measures corresponding to the nodes in all of the networks presented in this survey. Specifically, the functions used for ranking the nodes according to the corresponding centrality method described in Section \ref{subsection-methods-centrality} are presented in Table \ref{table-networkx-centrality}.

\begin{table}
    \caption{The \textit{NetworkX} functions used for computing the centrality methods.}
    \label{table-networkx-centrality}
    \begin{center}
          \begin{tabular}{|l|l|c|}
               \hline
               \textbf{Method} & \textbf{Function} & \textbf{Parameters} \\
               \hline
               Degree centrality & \textit{in\_degree\_centrality} & - \\
               Closeness centrality & \textit{closeness\_centrality} & default \\
               Harmonic centrality & \textit{harmonic\_centrality} & default \\
               Eccentricity centrality & \textit{eccentricity} & default \\
               Betweenness centrality & \textit{betweenness\_centrality} & default \\
               Eigenvector-based prestige & \textit{eigenvector\_centrality} & default \\
               \hline
          \end{tabular}
    \end{center}
\end{table}

Additionally, we used the \textit{dominating\_set} function with the default parameters to run the MDS analysis described in Section \ref{subsection-methods-controllability} for each network. While the package does not provide a direct equivalent for the MkDS algorithm, it allows for the initial network manipulation (i.e. parsing the network and adding specific edges) required to transform the default MDS analysis into an MkDS one. The methodology and results are presented in detail in Section~\ref{subsection-applications-results}.

\subsubsection{Cytoscape}
\label{subsubsection-methods-software-cytoscape}

\textit{Cytoscape} (\cite{cytoscape-article}) is a standalone software platform implemented in Java for visualizing complex networks, together with attribute data integration (\cite{cytoscape-url}).  \textit{Cytoscape} was initially developed for biological research and biological networks visualization and analysis. In addition to the basic functionality, \textit{Cytoscape} provides extended functionality through the use of so-called apps, which can be community-developed and add support for additional analysis methods, layouts, or database connections, among others.

Another well-established desktop network visualization software is \textit{Gephi} (\cite{gephi-url}), a Java application with the same capabilities, including the usage of plugins for extended functionality.

We have used the \textit{Cytoscape} application to render all of the network vi\-su\-ali\-za\-tions presented throughout this survey. The layout of the nodes was automatically calculated and applied using the default preferred layout algorithm, while the general design of the networks (e.g. labels, colors, or sizes) is based on the default style. A screenshot with the visualization of a network, can be seen in Figure \ref{figure-cytoscape-screenshot}.

\begin{figure}
     \begin{center}
          \includegraphics[width=0.9\textwidth]{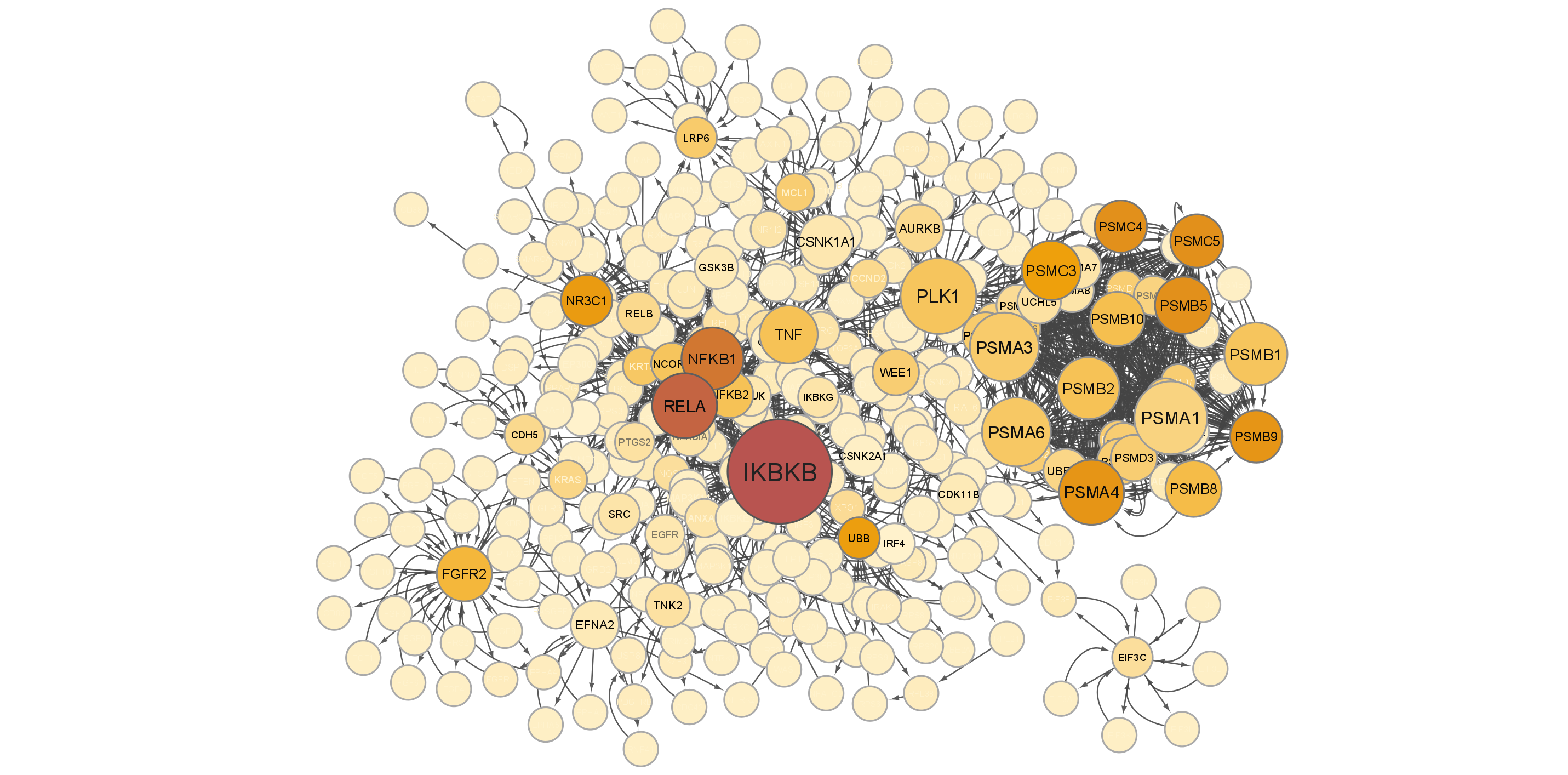}
     \end{center}
     \caption{Example of network visualization with \emph{Cytoscape}. The color of a node is proportional to its indegree (darker nodes have higher indegree), while the size of a node is proportional to its outdegree (larger nodes have higher outdegree).}
     \label{figure-cytoscape-screenshot}
\end{figure}

\subsubsection{NetControl4BioMed}
\label{subsubsection-methods-software-netcontrol4biomed}

\textit{NetControl4BioMed} (\cite{netcontrol4biomed-article}, updated version in \cite{NetControl4BioMed-v2021}) is a C\# (.NET Core) web application for the generation and structural target controllability analysis of protein-protein interaction networks, freely available at \url{https://netcontrol.combio.org/}. To this end, the application integrates and combines multiple protein and protein-protein interaction databases which, based on the user input, are used in the process.

We have used the \textit{NetControl4BioMed} platform to generate the personalized protein-protein interaction networks presented in this survey. The details of the network generation method (e.g. used interaction databases, or algorithm parameters) are presented in Section \ref{subsection-applications-data}. A screenshot with the home page of the application can be seen in Figure \ref{figure-netcontrol4biomed-screenshot}.

\begin{figure}
     \begin{center}
          \includegraphics[width=0.9\textwidth]{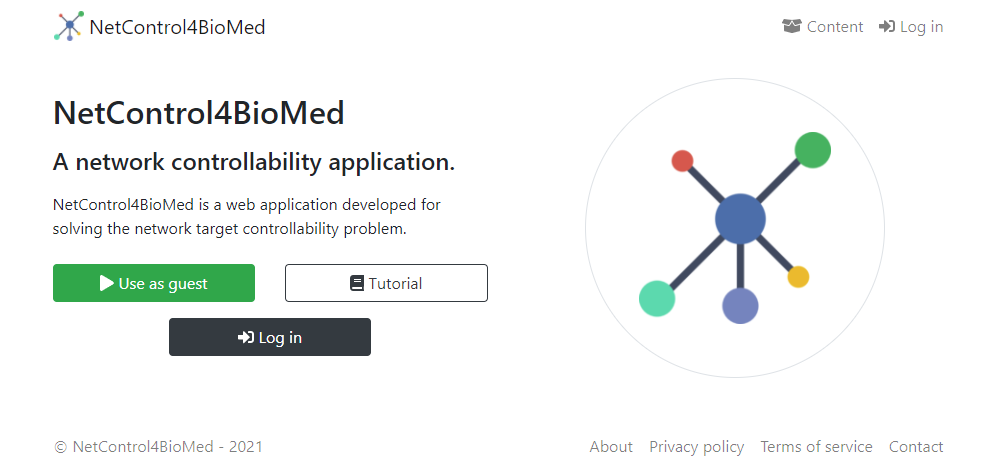}
     \end{center}
     \caption{Screenshot with the home page of \textit{NetControl4BioMed}.}
     \label{figure-netcontrol4biomed-screenshot}
\end{figure}

Additionally, we used the network analysis section of the application, with the default parameters, to run the structural target controllability analysis described in Section \ref{subsection-methods-controllability} for each network. The methodology and results are presented in detail in Section \ref{subsection-applications-results}.

%% file: section-applications.tex
\section{Applications of network modeling in personalized medicine}
\label{section-applications}

We demonstrate our approach to personalized medicine on three multiple my\-e\-lo\-ma patients, through the analysis of customized networks built around the mutated genes of each patient, the disease-specific survivability-essential genes, and the genes targeted by drugs in the standard therapy for multiple myeloma.

\input{subsection-applications-data}

\input{subsection-applications-methods}

\input{subsection-applications-results}

%% file: subsection-applications-data.tex
\subsection{Constructing personalized disease networks}
\label{subsection-applications-data}

We used the patient data documented in \cite{Lohr:2014aa}, which includes information about the evolution of the mutation, the mutated genes, and the stage of treatment, as well as details about the patient characteristics, such as age, race, and gender, among others. In this study we focus on the mutated genetic information for tumor samples 28, 38 and 191.

We used a list of 70 multiple myeloma-specific essential genes presented in \cite{essentialGenesGeoffrey, essentialGenesRodger, essentialGenesJan} and shown in Table \ref{essentialGenesTable}.

\begin{table}[htb]
    \caption{The disease-specific survivability-essential genes for multiple myeloma (\cite{essentialGenesGeoffrey, essentialGenesRodger, essentialGenesJan}).}
    \label{essentialGenesTable}
    \begin{center}
        \begin{tabular}{|llllll|}
            \hline
			AGTRAP & EIF3C & KIFC2 & PLK1 & RGAG1 & TRIM68 \\
            AURKB & EIF4A3 & LEPROT & PRPF8 & RPL27 & TUBGCP6 \\
            CARS & GNRH2 & MAF & PSMA1 & RPL38 & UBB \\
            CCND2 & GPR77 & MCL1 & PSMA3 & RRM1 & UBQLNL \\
            CDK11 & HIP1 & MED14 & PSMA4 & RSF1 & ULK3 \\
            CDK11A & IK & MED15 & PSMA6 & SF3A1 & USP36 \\
            CDK11B & IKBKB & NDC80 & PSMC3 & SLC25A23 & USP8 \\
            CKAP5 & IKZF1 & NFKB1 & PSMC4 & SNRPA1 & WBSCR22 \\
            COPB2 & IKZF3 & NFKB2 & PSMC5 & SNW1 & WEE1 \\
            CSNK1A1 & IRF4 & NUF2 & RAB11A & TNK2 & XPO1 \\
            CUL9 & KIF11 & PCDH18 & RELA & TPMT &  \\
            EFNA2 & KIF18A & PIM2 & RELB & TRIM21 &  \\
            \hline
        \end{tabular}
    \end{center}
\end{table}

We used the multiple myeloma standard treatment drugs described in \cite{mmFoundation, Engelhardt2019, 10.1056/NEJMoa1714678} and their corresponding drug-targets from DrugBank \cite{drugbank}. The list thus obtained is presented in Table \ref{drugTargetsTable} while the associated drugs and the standard drug therapies are documented in Tables~\ref{standardDrugsTable} and~\ref{standardTreatmentsTable}, respectively.

\begin{table}[htb]
    \caption{The targets of the drugs used in standard therapy lines for multiple myeloma.}
    \label{drugTargetsTable}
    \begin{center}
        \begin{tabular}{|llllll|}
            \hline
			ANXA1 & GSR & NR0B1 & PSMB2 & SLAMF7 & TUBB \\
            CD38 & HSD11B1 & NR1I2 & PSMB5 & TNF & XPO1 \\
            CDH5 & NFKB1 & NR3C1 & PSMB8 & TNFSF11 &  \\
            CRBN & NOLC1 & PSMB1 & PSMB9 & TOP2A &  \\
            FGFR2 & NOS2 & PSMB10 & PTGS2 & TUBA4A &  \\
            \hline
        \end{tabular}
    \end{center}
\end{table}

\begin{table}[htb]
    \caption{The main drugs for treating multiple myeloma.}
    \label{standardDrugsTable}
    \begin{center}
        \begin{tabular}{|p{0.3\textwidth}|p{0.6\textwidth}|}
            \hline
            \textbf{Drug} & \textbf{Drug-targets} \\
            \hline
			Bortezomib & PSMB1, PSMB5 \\
            Carfilzomib & PSMB1, PSMB10, PSMB2, PSMB5, PSMB8, PSMB9 \\
            Carmustine & GSR \\
            Cisplatin & A2M, ATOX1, MPG, TF \\
            Cyclophosphamide & NR1I2 \\
            Dacetuzumab & CD40 \\
            Daratumumab & CD38 \\
            Dexamethasone & ANXA1, NOS2, NR0B1, NR1I2, NR3C1 \\
            Doxorubicin & NOLC1, TOP2A \\
            Elotuzumab & SLAMF7 \\
            Etoposide & TOP2A, TOP2B \\
            Ixazomib & PSMB5 \\
            Lenalidomide & CDH5, CRBN, PTGS2, TNFSF11 \\
            Liposomal doxorubicin & TOP2A, TOP2B \\
            Oprozomib & LMP7, PSMB5 \\
            Panobinostat & HDAC1, HDAC2, HDAC3, HDAC6, HDAC7, HIF1A, VEGF \\
            Plerixafor & CXCR4 \\
            Pomalidomide & CRBN, PTGS2, TNF \\
            Prednisone & HSD11B1, NR3C1 \\
            Selinexor & XPO1 \\
            Thalidomide & CRBN, FGFR2, NFKB1, PTGS2, TNF \\
            Vincristine & TUBA4A, TUBB \\
            \hline
        \end{tabular}
    \end{center}
\end{table}

\begin{table}[htb]
    \caption{Lines of therapy for treating multiple myeloma.}
    \label{standardTreatmentsTable}
    \begin{center}
        \begin{tabular}{|p{0.9\textwidth}|}
            \hline
            \textbf{First line of therapy} \\
            \hline
            Lenalidomide, Bortezomib, Dexamethasone (RVD) \\
            Bortezomib, Cyclophosphamide, Dexamethasone (VKD) \\
            Bortezomib, Thalidomide, Dexamethasone (VTD) \\
            Bortezomib, Melphalan, Prednisone \\
            Vincristine, Doxorucibin, Dexamethasone (VAD) \\
            Melphalan, Dexamethasone \\
            Daratumumab, Bortezomib, Thalidomide, Dexamethasone \\
            Carfilzomib, Thalidomide, Dexamethasone (KTD) \\
            \hline
            \textbf{Second line of therapy} \\
            \hline
            Carfilzomib, Lenalidomide, Dexamethasone (KRD) \\
            Ixazomib, Lenalidomide, Dexamethasone \\
            Elotumuzab, Lenalidomide, Dexamethasone \\
            Bendamustine, Lenalidomide, Dexamethasone \\
            \hline
            \textbf{Third line of therapy} \\
            \hline
            Pomalidomide, Dexamethasone \\
            Panobinostat, Bortezomib, Dexamethasone \\
            Daratumumab \\
            \hline
        \end{tabular}
    \end{center}
\end{table}

We used the \textit{NetControl4BioMed} application, which was briefly presented in Section~\ref{subsubsection-methods-software-netcontrol4biomed}, to build a personalized protein-protein interaction network for each multiple myeloma patient around the seed genes defined by the patient-specific mutated genes, the disease-specific survivability-essential genes, and the drug-target genes corresponding to the standard treatment drugs. We used the interaction data from the \emph{KEGG}, \emph{OmniPath}, \emph{InnateDB} and \emph{SIGNOR} databases. The networks include all paths of length at most three between the seed proteins that could be formed with these interactions. We added to the network all intermediary nodes that were not part of the set of seed nodes. An overview of the generated networks is presented in Table \ref{table-results-networks}.

\begin{table}[htb]
    \caption{The summary of the generated personalized protein-protein interaction networks.}
    \label{table-results-networks}
    \begin{center}
        \begin{tabular}{|lcccccc|}
            \hline
			\textbf{Network} & \textbf{G} & \textbf{N} & \textbf{E} & \textbf{CC} & \textbf{D} & \textbf{AD} \\
			\hline
			Tumor sample 28 & $36$ & $360$ & $1486$ & $4$ & $12$ & $6.02$ \\
			Tumor sample 38 & $117$ & $446$ & $1732$ & $5$ & $12$ & $5.77$ \\
			Tumor sample 191 & $218$ & $515$ & $1955$ & $3$ & $12$ & $5.74$ \\
            \hline
        \end{tabular}
    \end{center}
    \small
    \caption*{\textbf{G}: number of mutated genes in the sample; \textbf{N}: number of nodes in the network; \textbf{E}: number of edges; \textbf{CC}: number of connected components; \textbf{D}: the network diameter; \textbf{AD}: the network average degree.}
\end{table}

We used the \textit{NetworkX} package, briefly presented in Section \ref{subsubsection-methods-software-networkx}, to rank the genes in each of the generated networks based on the centrality measures described in Section \ref{subsection-methods-centrality}. Tables \ref{table-results-centrality-in-degree}, \ref{table-results-centrality-closeness}, \ref{table-results-centrality-harmonic}, \ref{table-results-centrality-eccentricity}, \ref{table-results-centrality-betweenness}, and \ref{table-results-centrality-eigenvector} present, for each network, the essential genes that have the corresponding centrality measure higher than the median among all other genes in the network.

\begin{table}
    \caption{The top ranked genes based on their in-degree centrality, for each network.}
    \label{table-results-centrality-in-degree}
    \begin{center}
		\small
        \begin{tabular}{|ll|ll|ll|}
            \hline
			\multicolumn{2}{|c|}{\textbf{Tumor sample 28}} & \multicolumn{2}{|c|}{\textbf{Tumor sample 38}} & \multicolumn{2}{|c|}{\textbf{Tumor sample 191}} \\
			\hline
			IKBKB & AURKB & IKBKB & CCND2 & RELA & CCND2 \\
            RELA & XPO1 & RELA & XPO1 & IKBKB & XPO1 \\
            NFKB1 & EIF3C & NFKB1 & EIF3C & NFKB1 & EIF3C \\
            PSMC4 & TNK2 & PSMC4 & TNK2 & PSMC4 & TNK2 \\
            PSMC5 & CSNK1A1 & PSMC5 & CSNK1A1 & PSMC5 & RPL38 \\
            PSMA4 & KIF11 & PSMA4 & KIF11 & UBB & CSNK1A1 \\
            UBB & EFNA2 & UBB & CDK11B & PSMA4 & SNW1 \\
            PSMC3 & CDK11B & PSMC3 & SNW1 & PSMC3 & CDK11B \\
            NFKB2 & SNW1 & NFKB2 & EFNA2 & NFKB2 & KIF11 \\
            PLK1 & USP8 & PLK1 & IKZF3 & PLK1 & IKZF3 \\
            PSMA6 & IKZF3 & MCL1 & MED14 & MCL1 & IKZF1 \\
            PSMA3 & RSF1 & PSMA6 & USP8 & PSMA6 & EFNA2 \\
            MCL1 & IKZF1 & WEE1 & RSF1 & WEE1 & RRM1 \\
            WEE1 & RPL38 & PSMA3 & IKZF1 & PSMA3 & USP8 \\
            PSMA1 & TRIM21 & RELB & RPL38 & RELB & TRIM21 \\
            CCND2 & MAF & PSMA1 & TRIM21 & PSMA1 & RSF1 \\
            RELB &  & AURKB &  & AURKB & RPL27 \\
            \hline
        \end{tabular}
    \end{center}
\end{table}

\begin{table}
    \caption{The top ranked genes based on their closeness centrality, for each network.}
    \label{table-results-centrality-closeness}
    \begin{center}
		\small
        \begin{tabular}{|ll|ll|ll|}
            \hline
			\multicolumn{2}{|c|}{\textbf{Tumor sample 28}} & \multicolumn{2}{|c|}{\textbf{Tumor sample 38}} & \multicolumn{2}{|c|}{\textbf{Tumor sample 191}} \\
			\hline
			UBB & CSNK1A1 & IKBKB & PSMA4 & RELA & PSMA4 \\
            IKBKB & PSMA4 & UBB & XPO1 & UBB & CSNK1A1 \\
            RELA & XPO1 & RELA & CSNK1A1 & IKBKB & PSMC3 \\
            NFKB1 & PLK1 & NFKB1 & PSMC3 & NFKB1 & XPO1 \\
            NFKB2 & SNW1 & NFKB2 & SNW1 & NFKB2 & SNW1 \\
            RELB & PSMC3 & RELB & TRIM21 & RELB & PSMC4 \\
            WEE1 & PSMC4 & PSMA3 & PIM2 & WEE1 & AURKB \\
            PSMA3 & PSMC5 & MCL1 & PSMC4 & PSMA3 & PSMC5 \\
            MCL1 & TRIM21 & WEE1 & PSMC5 & MCL1 & PIM2 \\
            RSF1 & KIF11 & RSF1 & PSMA1 & RSF1 & PSMA1 \\
            CCND2 & PSMA1 & CCND2 & KIF11 & CCND2 & TRIM21 \\
            PSMA6 & PIM2 & PLK1 & TNK2 & PSMA6 & KIF11 \\
             &  & PSMA6 & AURKB & PLK1 & TNK2 \\
            \hline
        \end{tabular}
    \end{center}
\end{table}

\begin{table}
    \caption{The top ranked genes based on their harmonic centrality, for each network.}
    \label{table-results-centrality-harmonic}
    \begin{center}
		\small
        \begin{tabular}{|ll|ll|ll|}
            \hline
			\multicolumn{2}{|c|}{\textbf{Tumor sample 28}} & \multicolumn{2}{|c|}{\textbf{Tumor sample 38}} & \multicolumn{2}{|c|}{\textbf{Tumor sample 191}} \\
			\hline
			IKBKB & PSMC3 & IKBKB & PSMC5 & RELA & PSMC3 \\
            RELA & PSMA6 & RELA & PLK1 & IKBKB & PLK1 \\
            NFKB1 & MCL1 & NFKB1 & CCND2 & NFKB1 & PSMA6 \\
            UBB & CCND2 & UBB & RSF1 & UBB & RSF1 \\
            NFKB2 & RSF1 & NFKB2 & XPO1 & NFKB2 & CCND2 \\
            PSMA4 & PLK1 & RELB & PSMA1 & RELB & PSMA1 \\
            RELB & XPO1 & PSMA4 & CSNK1A1 & PSMA4 & AURKB \\
            PSMA3 & PSMA1 & PSMA3 & SNW1 & WEE1 & CSNK1A1 \\
            WEE1 & CSNK1A1 & WEE1 & TRIM21 & PSMA3 & XPO1 \\
            PSMC4 & SNW1 & MCL1 & PIM2 & PSMC4 & SNW1 \\
            PSMC5 &  & PSMC3 & AURKB & MCL1 & PIM2 \\
             &  & PSMA6 & TNK2 & PSMC5 & TRIM21 \\
             &  & PSMC4 & KIF11 &  &  \\
            \hline
        \end{tabular}
    \end{center}
\end{table}

\begin{table}
    \caption{The top ranked genes based on their eccentricity, for each network.}
    \label{table-results-centrality-eccentricity}
    \begin{center}
		\small
        \begin{tabular}{|ll|ll|ll|}
            \hline
			\multicolumn{2}{|c|}{\textbf{Tumor sample 28}} & \multicolumn{2}{|c|}{\textbf{Tumor sample 38}} & \multicolumn{2}{|c|}{\textbf{Tumor sample 191}} \\
			\hline
			IKBKB & RELB & IKBKB & EFNA2 & IKBKB & EFNA2 \\
            RELA & MCL1 & RELA & RELB & RELA & MCL1 \\
            NFKB1 & EFNA2 & NFKB1 & MCL1 & NFKB1 & TNK2 \\
            PSMA4 & TNK2 & PLK1 & TNK2 & PLK1 & CCND2 \\
            PLK1 & CCND2 & PSMA4 & CCND2 & PSMA4 & EIF3C \\
            PSMC3 & EIF3C & PSMC3 & EIF3C & PSMC5 & CDK11B \\
            PSMC5 & CDK11B & PSMC5 & CDK11B & PSMC4 & XPO1 \\
            PSMC4 & XPO1 & PSMC4 & XPO1 & PSMC3 & SNW1 \\
            PSMA6 & IRF4 & PSMA6 & IRF4 & PSMA6 & RPL38 \\
            PSMA1 & SNW1 & PSMA1 & SNW1 & PSMA1 & IRF4 \\
            PSMA3 & USP8 & PSMA3 & PIM2 & PSMA3 & PIM2 \\
            UBB & PIM2 & UBB & MED14 & UBB & USP8 \\
            NFKB2 & HIP1 & NFKB2 & USP8 & NFKB2 & IKZF1 \\
            WEE1 & KIF11 & WEE1 & RSF1 & CSNK1A1 & HIP1 \\
            CSNK1A1 & RSF1 & AURKB & HIP1 & WEE1 & KIF11 \\
            AURKB &  & CSNK1A1 & KIF11 & AURKB & IKZF3 \\
             &  &  &  & RELB & RAB11A \\
            \hline
        \end{tabular}
    \end{center}
\end{table}

\begin{table}
    \caption{The top ranked genes based on their betweenness centrality, for each network.}
    \label{table-results-centrality-betweenness}
    \begin{center}
		\small
        \begin{tabular}{|ll|ll|ll|}
            \hline
			\multicolumn{2}{|c|}{\textbf{Tumor sample 28}} & \multicolumn{2}{|c|}{\textbf{Tumor sample 38}} & \multicolumn{2}{|c|}{\textbf{Tumor sample 191}} \\
			\hline
			IKBKB & PSMA6 & IKBKB & RELB & IKBKB & RELB \\
            NFKB1 & EIF3C & RELA & CCND2 & RELA & UBB \\
            RELA & RELB & NFKB1 & EIF3C & NFKB1 & PSMA6 \\
            PLK1 & PSMA1 & PLK1 & PSMA1 & PLK1 & PSMA1 \\
            CSNK1A1 & PSMA4 & PSMC3 & PSMA4 & NFKB2 & PSMC5 \\
            NFKB2 & PSMC5 & NFKB2 & PSMC5 & WEE1 & PSMA4 \\
            PSMA3 & EFNA2 & WEE1 & EFNA2 & CSNK1A1 & EFNA2 \\
            WEE1 & USP8 & PSMA3 & MED14 & CDK11B & PSMC4 \\
            PSMC3 & PIM2 & CDK11B & USP8 & PSMA3 & RAB11A \\
            CDK11B & PSMC4 & CSNK1A1 & PIM2 & MCL1 & USP8 \\
            MCL1 & MED14 & TNK2 & PSMC4 & AURKB & MED14 \\
            TNK2 & RPL38 & AURKB & RSF1 & EIF3C & MAF \\
            AURKB & MAF & MCL1 & SNW1 & PSMC3 & PIM2 \\
            XPO1 & RSF1 & XPO1 & RPL38 & TNK2 & SNW1 \\
            CCND2 & SNW1 & PSMA6 & MAF & XPO1 & RPL38 \\
            UBB &  & UBB &  & CCND2 & RSF1 \\
            \hline
        \end{tabular}
    \end{center}
\end{table}

\begin{table}
    \caption{The top ranked genes based on their eigenvector-based prestige, for each network.}
    \label{table-results-centrality-eigenvector}
    \begin{center}
		\small
        \begin{tabular}{|ll|ll|ll|}
            \hline
			\multicolumn{2}{|c|}{\textbf{Tumor sample 28}} & \multicolumn{2}{|c|}{\textbf{Tumor sample 38}} & \multicolumn{2}{|c|}{\textbf{Tumor sample 191}} \\
			\hline
			PSMC4 & NFKB1 & PSMC4 & NFKB2 & PSMC4 & CCND2 \\
            PSMC5 & CCND2 & PSMC5 & CCND2 & PSMC5 & WEE1 \\
            PSMA4 & NFKB2 & PSMA4 & WEE1 & PSMA4 & MCL1 \\
            PSMC3 & WEE1 & PSMC3 & MCL1 & PSMC3 & RELB \\
            PSMA6 & MCL1 & PSMA6 & RELB & PSMA6 & RSF1 \\
            PSMA3 & RELB & PSMA3 & PLK1 & PSMA3 & AURKB \\
            PSMA1 & RSF1 & PSMA1 & RSF1 & PSMA1 & PLK1 \\
            UBB & SNW1 & UBB & SNW1 & UBB & SNW1 \\
            IKBKB & XPO1 & IKBKB & XPO1 & IKBKB & XPO1 \\
            RELA & PLK1 & RELA & AURKB & RELA & KIF11 \\
             &  & NFKB1 &  & NFKB1 & RPL27 \\
             &  &  &  & NFKB2 & PIM2 \\
            \hline
        \end{tabular}
    \end{center}
\end{table}

%% file: subsection-applications-methods.tex
\subsection{Analysis methods}
\label{subsection-applications-methods}

In this section, we describe the methodology applied for the analysis of the data presented in Section \ref{subsection-applications-data}, and using the centrality measures presented in Section \ref{subsection-methods-centrality} and the controllability methods presented in Section \ref{subsection-methods-controllability}.

All analyses follow a similar flow, aiming to identify, through the different controllability methods, a subset of ``important'' drug-target genes in the network, ranked according to the number of essential genes that they control. Once a set has been identified, the corresponding drugs are ranked based on a very similar criterion, taking into account the number of essential genes that their drug-targets control. The three top ranked drugs are then reported as a personalized drug combination therapy customized to the patient, and compared with the standard lines of therapy.

Firstly, we ran each type of analysis on the complete networks and data sets, considering as targets the essential genes presented in Table \ref{essentialGenesTable} and as preferred inputs the drug-target genes in Table \ref{drugTargetsTable}. Then, to reduce the noise in the data, for each network and each centrality measure we focused on the subgraphs formed by the top ranked essential genes and the drug-targets that can reach them through a path of length of $3$ or less, and all the interactions between them. For each such subgraph, all analyses consider as targets the corresponding essential genes, and as preferred inputs the corresponding drug-target genes.

We used the \textit{NetControl4BioMed} application, briefly presented in Section \ref{subsubsection-methods-software-netcontrol4biomed}, to run the structural target controllability analyses with the previously described setup. We used the default parameters, with a set maximum path length of $3$. The analysis outputs a set of genes that can control the entire target set, from which only the controlling drug-targets are selected and further considered.

Similarly, we used the \textit{NetworkX} package, presented in Section~\ref{subsubsection-methods-software-networkx}, to run the minimum dominating set analyses with the previously described setup. An additional step was required in order to enable the default function to perform the required minimum $k$-dominating set analysis. To this end, we transformed the analyzed networks by adding a direct edge between each pair of nodes indirectly connected by a path of length of $k$ or less. We set $k = 3$, and used the default parameters. The analysis outputs a set of genes that can dominate the entire target set, from which only the dominating drug-targets are selected and further considered.

%% file: subsection-applications-results.tex
\subsection{Results}
\label{subsection-applications-results}

In this section, we present and discuss the results of the analyses described in Section \ref{subsection-applications-methods}, with the aim of suggesting personalized treatments.

\subsubsection{Structural controllability analysis}
\label{subsection-applications-results-stc}

We applied the input-constrained structural target controllability method on the complete networks and sets of essential and drug-target genes. The drug-target genes in the obtained controlling sets are presented in Table \ref{table-results-stc-all}.

\begin{table}
    \caption{The drug-target genes in the controlling set obtained by the structural target controllability analysis.}
	\label{table-results-stc-all}
	\begin{center}
		\small
        \begin{tabular}{|ll|ll|ll|}
            \hline
			\multicolumn{2}{|c|}{\textbf{Tumor sample 28}} & \multicolumn{2}{|c|}{\textbf{Tumor sample 38}} & \multicolumn{2}{|c|}{\textbf{Tumor sample 191}} \\
			\hline
            ANXA1 & NOLC1 & ANXA1 & PSMB1 & ANXA1 & PSMB2 \\
            NFKB1 &  & NOLC1 & TNF &  &  \\
            \hline
        \end{tabular}
	\end{center}
\end{table}

Next, we applied the method once more on the subgraphs described in Section \ref{subsection-applications-methods} and corresponding to each centrality measure. The drug-target genes in the controlling sets obtained for each measure are presented in Tables~\ref{table-results-stc-in-degree}, \ref{table-results-stc-closeness}, \ref{table-results-stc-harmonic}, \ref{table-results-stc-eccentricity}, \ref{table-results-stc-betweenness} and \ref{table-results-stc-eigenvector}.

\begin{table}
    \caption{The drug-target genes in the controlling set obtained by the structural target controllability analysis corresponding to the top ranked essential genes based on their in-degree centrality.}
	\label{table-results-stc-in-degree}
	\begin{center}
		\small
        \begin{tabular}{|ll|ll|ll|}
            \hline
			\multicolumn{2}{|c|}{\textbf{Tumor sample 28}} & \multicolumn{2}{|c|}{\textbf{Tumor sample 38}} & \multicolumn{2}{|c|}{\textbf{Tumor sample 191}} \\
			\hline
                ANXA1 & NOLC1 & ANXA1 & TNF & ANXA1 & NOLC1 \\
                NFKB1 & TNF & NFKB1 &  & NFKB1 & TNF \\
            \hline
        \end{tabular}
	\end{center}
\end{table}

\begin{table}
    \caption{The drug-target genes in the controlling set obtained by the structural target controllability analysis corresponding to the top ranked essential genes based on their closeness centrality.}
	\label{table-results-stc-closeness}
	\begin{center}
		\small
        \begin{tabular}{|ll|ll|ll|}
            \hline
			\multicolumn{2}{|c|}{\textbf{Tumor sample 28}} & \multicolumn{2}{|c|}{\textbf{Tumor sample 38}} & \multicolumn{2}{|c|}{\textbf{Tumor sample 191}} \\
			\hline
                ANXA1 & PSMB8 & FGFR2 & PSMB5 & ANXA1 & PSMB1 \\
                NFKB1 & TNF & NFKB1 & TNF & NFKB1 & TNF \\
                 &  & NR3C1 &  & NR3C1 & XPO1 \\
            \hline
        \end{tabular}
	\end{center}
\end{table}

\begin{table}
    \caption{The drug-target genes in the controlling set obtained by the structural target controllability analysis corresponding to the top ranked essential genes based on their harmonic centrality.}
	\label{table-results-stc-harmonic}
	\begin{center}
		\small
        \begin{tabular}{|ll|ll|ll|}
            \hline
			\multicolumn{2}{|c|}{\textbf{Tumor sample 28}} & \multicolumn{2}{|c|}{\textbf{Tumor sample 38}} & \multicolumn{2}{|c|}{\textbf{Tumor sample 191}} \\
			\hline
                ANXA1 & PTGS2 & ANXA1 & TNF & ANXA1 & PSMB1 \\
                NFKB1 & TNF & NFKB1 & XPO1 & NFKB1 & TNF \\
                PSMB1 & XPO1 & PSMB10 &  &  &  \\        
            \hline
        \end{tabular}
	\end{center}
\end{table}

\begin{table}
    \caption{The drug-target genes in the controlling set obtained by the structural target controllability analysis corresponding to the top ranked essential genes based on their eccentricity.}
	\label{table-results-stc-eccentricity}
	\begin{center}
		\small
        \begin{tabular}{|ll|ll|ll|}
            \hline
			\multicolumn{2}{|c|}{\textbf{Tumor sample 28}} & \multicolumn{2}{|c|}{\textbf{Tumor sample 38}} & \multicolumn{2}{|c|}{\textbf{Tumor sample 191}} \\
			\hline
                ANXA1 & PSMB1 & ANXA1 & NR3C1 & ANXA1 & NFKB1 \\
                NFKB1 & TNF & NFKB1 &  &  &  \\
            \hline
        \end{tabular}
	\end{center}
\end{table}

\begin{table}
    \caption{The drug-target genes in the controlling set obtained by the structural target controllability analysis corresponding to the top ranked essential genes based on their betweenness centrality.}
	\label{table-results-stc-betweenness}
	\begin{center}
		\small
        \begin{tabular}{|ll|ll|ll|}
            \hline
			\multicolumn{2}{|c|}{\textbf{Tumor sample 28}} & \multicolumn{2}{|c|}{\textbf{Tumor sample 38}} & \multicolumn{2}{|c|}{\textbf{Tumor sample 191}} \\
			\hline
                ANXA1 & TNF & ANXA1 & PTGS2 & ANXA1 & PTGS2 \\
                NFKB1 &  & NFKB1 & TNF & NFKB1 & TNF \\
                 &  & PSMB10 &  &  &  \\
            \hline
        \end{tabular}
	\end{center}
\end{table}

\begin{table}
    \caption{The drug-target genes in the controlling set obtained by the structural target controllability analysis corresponding to the top ranked essential genes based on their eigenvector-based prestige.}
	\label{table-results-stc-eigenvector}
	\begin{center}
		\small
        \begin{tabular}{|ll|ll|ll|}
            \hline
			\multicolumn{2}{|c|}{\textbf{Tumor sample 28}} & \multicolumn{2}{|c|}{\textbf{Tumor sample 38}} & \multicolumn{2}{|c|}{\textbf{Tumor sample 191}} \\
			\hline
                NFKB1 & XPO1 & ANXA1 & PSMB5 & ANXA1 & TNF \\
                PSMB1 &  & NFKB1 & TNF & NFKB1 &  \\
                 &  & PSMB2 &  &  &  \\           
            \hline
        \end{tabular}
	\end{center}
\end{table}

Next, we studied the results in the context of drug therapy. For each network and centrality measure, the drug-target genes were matched with all the multiple myeloma standard treatment drug targeting them, while the latter was sorted according to the number of top ranked essential genes controlled through one or more of their drug-target genes. Then, we selected the three top drugs as our proposed drug combination for that patient's treatment. Informally, this approach aims for a reinforced influential effect over as many essential genes as possible through a minimal combination of at most 3 drugs. The chosen drugs are documented in Tables \ref{table-drugs-stc-in-degree}, \ref{table-drugs-stc-closeness}, \ref{table-drugs-stc-harmonic}, \ref{table-drugs-stc-eccentricity}, \ref{table-drugs-stc-betweenness} and \ref{table-drugs-stc-eigenvector}. Following an identical procedure, we also suggest a personalized drug combination in Table \ref{table-drugs-stc-all} based on the conclusions reached by the controllability analysis on the whole tumor network. The associated most effective standard drug therapy for each tumor sample and centrality measure are given in Tables \ref{table-drugs-stc-standardTherapy-in-degree}, \ref{table-drugs-stc-standardTherapy-closeness}, \ref{table-drugs-stc-standardTherapy-harmonic}, \ref{table-drugs-stc-standardTherapy-eccentricity}, \ref{table-drugs-stc-standardTherapy-betweenness}, \ref{table-drugs-stc-standardTherapy-eigenvector} and \ref{table-drugs-stc-standardTherapy-all}.

\begin{table}
	\caption{The proposed drug therapy and the number of essential genes it controls, for the structural target controllability analysis corresponding to the top ranked essential genes based on their in-degree.}
	\label{table-drugs-stc-in-degree}
	\begin{center}
		\small 
		\begin{tabular}{|l|l|c|}
			\hline
			\textbf{Tumor Sample} & \textbf{Proposed therapy} & \textbf{Controlled EG} \\
			\hline
			MM-0028-Tumor & Dexamethasone, Thalidomide, Doxorubicin & 29 \\
			MM-0038-Tumor & Dexamethasone, Thalidomide, Pomalidomide & 24 \\
			MM-0191-Tumor & Dexamethasone, Thalidomide, Pomalidomide & 28 \\
			\hline
		\end{tabular}
	\end{center}
\end{table}

\begin{table}
    \caption{The most effective standard drug therapy and the number of essential genes it controls, for the structural target controllability analysis corresponding to the top ranked essential genes based on their in-degree.}
	\label{table-drugs-stc-standardTherapy-in-degree}
	\begin{center}
		\small
        \begin{tabular}{|l|l|c|}
            \hline
			\textbf{Tumor Sample} & \textbf{Standard therapy} & \textbf{Controlled EG} \\
			\hline
				MM-0028-Tumor & Bortezomib, Thalidomide, Dexamethasone & 27 \\
				MM-0038-Tumor & Bortezomib, Thalidomide, Dexamethasone & 24 \\
				MM-0191-Tumor & Bortezomib, Thalidomide, Dexamethasone & 28 \\
            \hline
        \end{tabular}
	\end{center}
\end{table}

\begin{table}
	\caption{The proposed drug therapy and the number of essential genes it controls, for structural target controllability analysis corresponding to the top ranked essential genes based on their closeness.}
	\label{table-drugs-stc-closeness}
	\begin{center}
		\small 
		\begin{tabular}{|l|l|c|}
			\hline
			\textbf{Tumor Sample} & \textbf{Proposed therapy} & \textbf{Controlled EG} \\
			\hline
			MM-0028-Tumor & Thalidomide, Dexamethasone, Pomalidomide & 21 \\
			MM-0038-Tumor & Thalidomide, Pomalidomide, Dexamethasone & 23 \\
			MM-0191-Tumor & Thalidomide, Dexamethasone, Selinexor & 24 \\
			\hline
		\end{tabular}
	\end{center}
\end{table}

\begin{table}
    \caption{The most effective standard drug therapy and the number of essential genes it controls, for the structural target controllability analysis corresponding to the top ranked essential genes based on their closeness.}
	\label{table-drugs-stc-standardTherapy-closeness}
	\begin{center}
		\small
        \begin{tabular}{|l|l|c|}
            \hline
			\textbf{Tumor Sample} & \textbf{Standard therapy} & \textbf{Controlled EG} \\
			\hline
				MM-0028-Tumor & Bortezomib, Thalidomide, Dexamethasone & 21 \\
				MM-0038-Tumor & Bortezomib, Thalidomide, Dexamethasone & 24 \\
				MM-0191-Tumor & Bortezomib, Thalidomide, Dexamethasone & 24 \\
            \hline
        \end{tabular}
	\end{center}
\end{table}

\begin{table}
	\caption{The proposed drug therapy and the number of essential genes it controls, for the structural target controllability analysis corresponding to the top ranked essential genes based on their harmonic centrality.}
	\label{table-drugs-stc-harmonic}
	\begin{center}
		\small 
		\begin{tabular}{|l|l|c|}
			\hline
			\textbf{Tumor Sample} & \textbf{Proposed therapy} & \textbf{Controlled EG} \\
			\hline
			MM-0028-Tumor & Thalidomide, Dexamethasone, Selinexor & 21 \\
			MM-0038-Tumor & Dexamethasone, Thalidomide, Pomalidomide & 22 \\
			MM-0191-Tumor & Dexamethasone, Thalidomide, Bortezomib & 22 \\
			\hline
		\end{tabular}
	\end{center}
\end{table}

\begin{table}
    \caption{The most effective standard drug therapy and the number of essential genes it controls, for the structural target controllability analysis corresponding to the top ranked essential genes based on their harmonic centrality.}
	\label{table-drugs-stc-standardTherapy-harmonic}
	\begin{center}
		\small
        \begin{tabular}{|l|l|c|}
            \hline
			\textbf{Tumor Sample} & \textbf{Standard therapy} & \textbf{Controlled EG} \\
			\hline
				MM-0028-Tumor & Bortezomib, Thalidomide, Dexamethasone & 20 \\
				MM-0038-Tumor & Bortezomib, Thalidomide, Dexamethasone & 22 \\
				MM-0191-Tumor & Bortezomib, Thalidomide, Dexamethasone & 22 \\
            \hline
        \end{tabular}
	\end{center}
\end{table}

\begin{table}
	\caption{The proposed drug therapy and the number of essential genes it controls, for the structural target controllability analysis corresponding to the top ranked essential genes based on their eccentricity.}
	\label{table-drugs-stc-eccentricity}
	\begin{center}
		\small 
		\begin{tabular}{|l|l|c|}
			\hline
			\textbf{Tumor Sample} & \textbf{Proposed therapy} & \textbf{Controlled EG} \\
			\hline
			MM-0028-Tumor & Thalidomide, Pomalidomide, Bortezomib & 25 \\
			MM-0038-Tumor & Dexamethasone, Prednisone, Thalidomide & 3 \\
			MM-0191-Tumor & Thalidomide, Dexamethasone & 26 \\
			\hline
		\end{tabular}
	\end{center}
\end{table}

\begin{table}
    \caption{The most effective standard drug therapy and the number of essential genes it controls, for the structural target controllability analysis corresponding to the top ranked essential genes based on their eccentricity.}
	\label{table-drugs-stc-standardTherapy-eccentricity}
	\begin{center}
		\small
        \begin{tabular}{|l|l|c|}
            \hline
			\textbf{Tumor Sample} & \textbf{Standard therapy} & \textbf{Controlled EG} \\
			\hline
				MM-0028-Tumor & Bortezomib, Thalidomide, Dexamethasone & 26 \\
				MM-0038-Tumor & Bortezomib, Thalidomide, Dexamethasone & 3 \\
				MM-0191-Tumor & Bortezomib, Thalidomide, Dexamethasone & 26 \\
            \hline
        \end{tabular}
	\end{center}
\end{table}

\begin{table}
	\caption{The proposed drug therapy and the number of essential genes it controls, for the structural target controllability analysis corresponding to the top ranked essential genes based on their betweenness.}
	\label{table-drugs-stc-betweenness}
	\begin{center}
		\small 
		\begin{tabular}{|l|l|c|}
			\hline
			\textbf{Tumor Sample} & \textbf{Proposed therapy} & \textbf{Controlled EG} \\
			\hline
			MM-0028-Tumor & Thalidomide, Pomalidomide, Dexamethasone & 26 \\
			MM-0038-Tumor & Thalidomide, Pomalidomide, Dexamethasone & 26 \\
			MM-0191-Tumor & Dexamethasone, Thalidomide, Pomalidomide & 27 \\
			\hline
		\end{tabular}
	\end{center}
\end{table}

\begin{table}
    \caption{The most effective standard drug therapy and the number of essential genes it controls, for the structural target controllability analysis corresponding to the top ranked essential genes based on their betweenness.}
	\label{table-drugs-stc-standardTherapy-betweenness}
	\begin{center}
		\small
        \begin{tabular}{|l|l|c|}
            \hline
			\textbf{Tumor Sample} & \textbf{Standard therapy} & \textbf{Controlled EG} \\
			\hline
				MM-0028-Tumor & Bortezomib, Thalidomide, Dexamethasone & 26 \\
				MM-0038-Tumor & Bortezomib, Thalidomide, Dexamethasone & 26 \\
				MM-0191-Tumor & Bortezomib, Thalidomide, Dexamethasone & 27 \\
            \hline
        \end{tabular}
	\end{center}
\end{table}

\begin{table}
	\caption{The proposed drug therapy and the number of essential genes it controls, for the structural target controllability analysis corresponding to the top ranked essential genes based on their eigenvector centrality.}
	\label{table-drugs-stc-eigenvector}
	\begin{center}
		\small 
		\begin{tabular}{|l|l|c|}
			\hline
			\textbf{Tumor Sample} & \textbf{Proposed therapy} & \textbf{Controlled EG} \\
			\hline
			MM-0028-Tumor & Thalidomide, Selinexor, Bortezomib & 3 \\
			MM-0038-Tumor & Dexamethasone, Thalidomide, Pomalidomide & 19 \\
			MM-0191-Tumor & Thalidomide, Dexamethasone, Pomalidomide & 21 \\
			\hline
		\end{tabular}
	\end{center}
\end{table}

\begin{table}
    \caption{The most effective standard drug therapy and the number of essential genes it controls, for the structural target controllability analysis corresponding to the top ranked essential genes based on their eigenvector centrality.}
	\label{table-drugs-stc-standardTherapy-eigenvector}
	\begin{center}
		\small
        \begin{tabular}{|l|l|c|}
            \hline
			\textbf{Tumor Sample} & \textbf{Standard therapy} & \textbf{Controlled EG} \\
			\hline
				MM-0028-Tumor & Bortezomib, Thalidomide, Dexamethasone & 2 \\
				MM-0038-Tumor & Bortezomib, Thalidomide, Dexamethasone & 19 \\
				MM-0191-Tumor & Bortezomib, Thalidomide, Dexamethasone & 21 \\
            \hline
        \end{tabular}
	\end{center}
\end{table}

\begin{table}
	\caption{The proposed drug therapy and the number of essential genes it controls, for the structural target controllability analysis corresponding to the essential genes in the patient network.}
	\label{table-drugs-stc-all}
	\begin{center}
		\small 
		\begin{tabular}{|l|l|c|}
			\hline
			\textbf{Tumor Sample} & \textbf{Proposed therapy} & \textbf{Controlled EG} \\
			\hline
			MM-0028-Tumor & Thalidomide, Doxorubicin, Dexamethasone & 31 \\
			MM-0038-Tumor & Thalidomide, Pomalidomide, Doxorubicin & 28 \\
			MM-0191-Tumor & Dexamethasone, Carfilzomib & 25 \\
			\hline
		\end{tabular}
	\end{center}
\end{table}

\begin{table}
    \caption{The most effective standard drug therapy and the number of essential genes it controls, for the structural target controllability analysis corresponding to the essential genes in the patient network.}
	\label{table-drugs-stc-standardTherapy-all}
	\begin{center}
		\small
        \begin{tabular}{|l|l|c|}
            \hline
			\textbf{Tumor Sample} & \textbf{Standard therapy} & \textbf{Controlled EG} \\
			\hline
				MM-0028-Tumor & Bortezomib, Thalidomide, Dexamethasone & 29 \\
				MM-0038-Tumor & Bortezomib, Thalidomide, Dexamethasone & 29 \\
				MM-0191-Tumor & Lenalidomide, Bortezomib, Dexamethasone & 24 \\
            \hline
        \end{tabular}
	\end{center}
\end{table}

It is immediate to see that Thalidomide and Dexamethasone are predicted to be extremely suitable for treating the unique disease circumstances afflicting all the patients. These two drugs are commonly preferred as a first choice when approaching multiple myeloma cases and are part of two therapies in combination with either  Bortezomib or Carfilzomib frequently use in latest medical practice, namely VTD and KTD respectively. These two combinations are supported by successful outcomes in different studies, such as \cite{ROUSSEL2020e874, KTDStudy2}. 

The third drug combination spot is usually taken by Pomalidomide. While this drug is mostly considered in later stages of the treatment in the context of traditional care, the analysis predicts this drug to have a strong reinforcement impact on the activity of Thalidomide for these specific cases, as it shares the same controlled essential genes with the latter. Our approach also identifies several outliers to the predominant 3-drug combination, namely Selinexor, Doxorubicin, Prednisone, Bortezomib, and Carfilzomib. The first two drugs are not considered until the last stages of the treatment in standard therapy lines. Although the conclusions collected by this study for these three patients align with this traditional approach, they still provide grounds to consider them over other late stage drugs if the treatment progresses to an evolved phase. On the other hand, the remaining three are relevant in starting medical diagnosis, with Bortezomib being specially remarkable. Following a more traditional approach, these three drugs may be given preference over Pomalidomide for the first prescriptions.

The suggested drug therapies are predicted to control most of the considered essential genes in the whole network and each of the centrality subnetworks. Consequently, it is expected for them to have a significant favorable impact on the condition of our targeted patients. Furthermore, the prescribed sequence of drugs are very close to the known therapies frequently used in the current state of the art treatment. All of these events provide a strong foundation for the feasibility of this method in personalized medicine.

\subsubsection{Minimum dominating set analysis}
\label{subsection-applications-results-m3ds}

We applied the input-constrained minimum 3-dominating set method on the complete networks and sets of essential and drug-target genes. The drug-target genes in the obtained dominating sets are shown in Table \ref{table-results-m3ds-all}.

\begin{table}
	\caption{The drug-target genes in the dominating set obtained by the minimum dominating set analysis.}
	\label{table-results-m3ds-all}
	\begin{center}
		\small
		\begin{tabular}{|ll|ll|ll|}
			\hline
			\multicolumn{2}{|c|}{\textbf{Tumor sample 28}} & \multicolumn{2}{|c|}{\textbf{Tumor sample 38}} & \multicolumn{2}{|c|}{\textbf{Tumor sample 191}}
			\\
			\hline
			ANXA1 & PSMB5 & ANXA1 & PSMB5 & ANXA1 & PSMB5 \\
            NFKB1 & PSMB9 & NFKB1 & PSMB8 & NFKB1 & PTGS2 \\
            NOLC1 & PTGS2 & NOLC1 & PSMB9 & NOLC1 & TNF \\
            NR3C1 & TNF & NR3C1 & PTGS2 & NR3C1 & TNFSF11 \\
            PSMB1 & TNFSF11 & PSMB1 & TNF &  &  \\
            PSMB2 & TOP2A & PSMB10 & TNFSF11 &  &  \\
             &  & PSMB2 & TOP2A &  &  \\
			\hline
		\end{tabular}
	\end{center}
\end{table}

Next, we applied the method once more on the subgraphs described in Section \ref{subsection-applications-methods} and corresponding to each centrality measure. The drug-target genes in the dominating set obtained for each measure are presented in Tables~\ref{table-results-m3ds-in-degree}, \ref{table-results-m3ds-closeness}, \ref{table-results-m3ds-harmonic}, \ref{table-results-m3ds-eccentricity}, \ref{table-results-m3ds-betweenness} and \ref{table-results-m3ds-eigenvector}.

\begin{table}
	\caption{The drug-target genes in the dominating set obtained by the minimum dominating set analysis corresponding to the top ranked essential genes based on their in-degree centrality.}
	\label{table-results-m3ds-in-degree}
	\begin{center}
		\small
		\begin{tabular}{|ll|ll|ll|}
			\hline
			\multicolumn{2}{|c|}{\textbf{Tumor sample 28}} & \multicolumn{2}{|c|}{\textbf{Tumor sample 38}} & \multicolumn{2}{|c|}{\textbf{Tumor sample 191}}
			\\
			\hline
			ANXA1 & PSMB5 & ANXA1 & NR3C1 & ANXA1 & NR3C1 \\
            NFKB1 & PSMB9 & NFKB1 & TNF & NFKB1 & TNF \\
            NOLC1 & PTGS2 & NOLC1 & TNFSF11 & NOLC1 & TNFSF11 \\
            NR3C1 & TNF &  &  &  &  \\
            PSMB1 & TNFSF11 &  &  &  &  \\
            PSMB2 & TOP2A &  &  &  &  \\
			\hline
		\end{tabular}
	\end{center}
\end{table}

\begin{table}
	\caption{The drug-target genes in the dominating set obtained by the minimum dominating set analysis corresponding to the top ranked essential genes based on their closeness centrality.}
	\label{table-results-m3ds-closeness}
	\begin{center}
		\small
		\begin{tabular}{|ll|ll|ll|}
			\hline
			\multicolumn{2}{|c|}{\textbf{Tumor sample 28}} & \multicolumn{2}{|c|}{\textbf{Tumor sample 38}} & \multicolumn{2}{|c|}{\textbf{Tumor sample 191}}
			\\
			\hline
			ANXA1 & TNF & ANXA1 & TNF & ANXA1 & TNF \\
            NFKB1 & TNFSF11 & NFKB1 & TNFSF11 & NFKB1 & TNFSF11 \\
            NR3C1 &  & NR3C1 &  & NR3C1 &  \\
			\hline
		\end{tabular}
	\end{center}
\end{table}

\begin{table}
	\caption{The drug-target genes in the dominating set obtained by the minimum dominating set analysis corresponding to the top ranked essential genes based on their harmonic centrality.}
	\label{table-results-m3ds-harmonic}
	\begin{center}
		\small
		\begin{tabular}{|ll|ll|ll|}
			\hline
			\multicolumn{2}{|c|}{\textbf{Tumor sample 28}} & \multicolumn{2}{|c|}{\textbf{Tumor sample 38}} & \multicolumn{2}{|c|}{\textbf{Tumor sample 191}}
			\\
			\hline
			ANXA1 & TNF & ANXA1 & TNF & ANXA1 & TNF \\
            NFKB1 &  & NFKB1 & TNFSF11 & NFKB1 & TNFSF11 \\
             &  & NR3C1 &  & NR3C1 &  \\
			\hline
		\end{tabular}
	\end{center}
\end{table}

\begin{table}
	\caption{The drug-target genes in the dominating set obtained by the minimum dominating set analysis corresponding to the top ranked essential genes based on their eccentricity.}
	\label{table-results-m3ds-eccentricity}
	\begin{center}
		\small
		\begin{tabular}{|ll|ll|ll|}
			\hline
			\multicolumn{2}{|c|}{\textbf{Tumor sample 28}} & \multicolumn{2}{|c|}{\textbf{Tumor sample 38}} & \multicolumn{2}{|c|}{\textbf{Tumor sample 191}}
			\\
			\hline
			ANXA1 & TNF & ANXA1 & TNF & ANXA1 & NR3C1 \\
            NFKB1 & TNFSF11 & NFKB1 & TNFSF11 & NFKB1 & TNF \\
            NR3C1 &  & NR3C1 &  & NOLC1 & TNFSF11 \\
			\hline
		\end{tabular}
	\end{center}
\end{table}

\begin{table}
	\caption{The drug-target genes in the dominating set obtained by the minimum dominating set analysis corresponding to the top ranked essential genes based on their betweenness centrality.}
	\label{table-results-m3ds-betweenness}
	\begin{center}
		\small
		\begin{tabular}{|ll|ll|ll|}
			\hline
			\multicolumn{2}{|c|}{\textbf{Tumor sample 28}} & \multicolumn{2}{|c|}{\textbf{Tumor sample 38}} & \multicolumn{2}{|c|}{\textbf{Tumor sample 191}}
			\\
			\hline
			ANXA1 & PSMB5 & ANXA1 & PTGS2 & ANXA1 & PTGS2 \\
            NFKB1 & PSMB9 & NFKB1 & TNF & NFKB1 & TNF \\
            NOLC1 & PTGS2 & NR3C1 & TNFSF11 & NR3C1 & TNFSF11 \\
            NR3C1 & TNF &  &  &  &  \\
            PSMB1 & TNFSF11 &  &  &  &  \\
            PSMB2 & TOP2A &  &  &  &  \\
			\hline
		\end{tabular}
	\end{center}
\end{table}

\begin{table}
	\caption{The drug-target genes in the dominating set obtained by the minimum dominating set analysis corresponding to the top ranked essential genes based on their eigenvector-based prestige.}
	\label{table-results-m3ds-eigenvector}
	\begin{center}
		\small
		\begin{tabular}{|ll|ll|ll|}
			\hline
			\multicolumn{2}{|c|}{\textbf{Tumor sample 28}} & \multicolumn{2}{|c|}{\textbf{Tumor sample 38}} & \multicolumn{2}{|c|}{\textbf{Tumor sample 191}}
			\\
			\hline
			ANXA1 & TNF & ANXA1 & TNF & ANXA1 & TNF \\
            NFKB1 &  & NFKB1 &  & NFKB1 &  \\
			\hline
		\end{tabular}
	\end{center}
\end{table}

The examination of these results follows an analogous procedure to the one introduced for target controllability. For each network and centrality measure, as well as the network associated to the whole minimum dominating set, we ascertain a drug therapy including the top ranked drugs on terms of reached essential genes. The chosen drugs are documented in Tables \ref{table-drugs-m3ds-in-degree}, \ref{table-drugs-m3ds-closeness}, \ref{table-drugs-m3ds-harmonic}, \ref{table-drugs-m3ds-eccentricity}, \ref{table-drugs-m3ds-betweenness} and \ref{table-drugs-m3ds-eigenvector}. Furthermore, the drugs corresponding to the dominating set associated to all the essential genes in the original patient network are shown in Table \ref{table-drugs-m3ds-all}. The associated most effective standard drug therapy for each tumor sample and centrality measure are given in Tables \ref{table-drugs-m3ds-standardTherapy-in-degree}, \ref{table-drugs-m3ds-standardTherapy-closeness}, \ref{table-drugs-m3ds-standardTherapy-harmonic}, \ref{table-drugs-m3ds-standardTherapy-eccentricity}, \ref{table-drugs-m3ds-standardTherapy-betweenness}, \ref{table-drugs-m3ds-standardTherapy-eigenvector} and \ref{table-drugs-m3ds-standardTherapy-all}.

\begin{table}
	\caption{The proposed drug therapy and the number of essential genes it dominates, for the minimum dominating set corresponding to the top ranked essential genes based on their in-degree.}
	\label{table-drugs-m3ds-in-degree}
	\begin{center}
		\small 
		\begin{tabular}{|l|l|c|}
			\hline
			\textbf{Tumor Sample} & \textbf{Proposed therapy} & \textbf{Dominated EG} \\
			\hline
			MM-0028-Tumor & Thalidomide, Pomalidomide, Dexamethasone & 28 \\
			MM-0038-Tumor & Thalidomide, Pomalidomide, Dexamethasone & 27 \\
			MM-0191-Tumor & Thalidomide, Pomalidomide, Dexamethasone & 29 \\
			\hline
		\end{tabular}
	\end{center}
\end{table}

\begin{table}
    \caption{The most effective standard drug therapy and the number of essential genes it dominates, for the minimum dominating set corresponding to the top ranked essential genes based on their in-degree.}
	\label{table-drugs-m3ds-standardTherapy-in-degree}
	\begin{center}
		\small
        \begin{tabular}{|l|l|c|}
            \hline
			\textbf{Tumor Sample} & \textbf{Standard therapy} & \textbf{Dominated EG} \\
			\hline
				MM-0028-Tumor & Bortezomib, Thalidomide, Dexamethasone & 28 \\
				MM-0038-Tumor & Bortezomib, Thalidomide, Dexamethasone & 27 \\
				MM-0191-Tumor & Bortezomib, Thalidomide, Dexamethasone & 29 \\
            \hline
        \end{tabular}
	\end{center}
\end{table}

\begin{table}
	\caption{The proposed drug therapy and the number of essential genes it dominates, for the minimum dominating set corresponding to the top ranked essential genes based on their closeness.}
	\label{table-drugs-m3ds-closeness}
	\begin{center}
		\small 
		\begin{tabular}{|l|l|c|}
			\hline
			\textbf{Tumor Sample} & \textbf{Proposed therapy} & \textbf{Dominated EG} \\
			\hline
			MM-0028-Tumor & Thalidomide, Pomalidomide, Dexamethasone & 24 \\
			MM-0038-Tumor & Thalidomide, Pomalidomide, Dexamethasone & 26 \\
			MM-0191-Tumor & Thalidomide, Pomalidomide, Dexamethasone & 26 \\
			\hline
		\end{tabular}
	\end{center}
\end{table}

\begin{table}
    \caption{The most effective standard drug therapy and the number of essential genes it dominates, for the minimum dominating set corresponding to the top ranked essential genes based on their closeness.}
	\label{table-drugs-m3ds-standardTherapy-closeness}
	\begin{center}
		\small
        \begin{tabular}{|l|l|c|}
            \hline
			\textbf{Tumor Sample} & \textbf{Standard therapy} & \textbf{Dominated EG} \\
			\hline
				MM-0028-Tumor & Bortezomib, Thalidomide, Dexamethasone & 24 \\
				MM-0038-Tumor & Bortezomib, Thalidomide, Dexamethasone & 26 \\
				MM-0191-Tumor & Bortezomib, Thalidomide, Dexamethasone & 26 \\
            \hline
        \end{tabular}
	\end{center}
\end{table}

\begin{table}
	\caption{The proposed drug therapy and the number of essential genes it dominates, for the minimum dominating set corresponding to the top ranked essential genes based on their harmonic centrality.}
	\label{table-drugs-m3ds-harmonic}
	\begin{center}
		\small 
		\begin{tabular}{|l|l|c|}
			\hline
			\textbf{Tumor Sample} & \textbf{Proposed therapy} & \textbf{Dominated EG} \\
			\hline
			MM-0028-Tumor & Thalidomide, Pomalidomide, Dexamethasone & 21 \\
			MM-0038-Tumor & Thalidomide, Pomalidomide, Dexamethasone & 26 \\
			MM-0191-Tumor & Thalidomide, Pomalidomide, Dexamethasone & 24 \\
			\hline
		\end{tabular}
	\end{center}
\end{table}

\begin{table}
    \caption{The most effective standard drug therapy and the number of essential genes it dominates, for the minimum dominating set corresponding to the top ranked essential genes based on their harmonic centrality.}
	\label{table-drugs-m3ds-standardTherapy-harmonic}
	\begin{center}
		\small
        \begin{tabular}{|l|l|c|}
            \hline
			\textbf{Tumor Sample} & \textbf{Standard therapy} & \textbf{Dominated EG} \\
			\hline
				MM-0028-Tumor & Bortezomib, Thalidomide, Dexamethasone & 21 \\
				MM-0038-Tumor & Bortezomib, Thalidomide, Dexamethasone & 26 \\
				MM-0191-Tumor & Bortezomib, Thalidomide, Dexamethasone & 24 \\
            \hline
        \end{tabular}
	\end{center}
\end{table}

\begin{table}
	\caption{The proposed drug therapy and the number of essential genes it dominates, for the minimum dominating set corresponding to the top ranked essential genes based on their eccentricity.}
	\label{table-drugs-m3ds-eccentricity}
	\begin{center}
		\small 
		\begin{tabular}{|l|l|c|}
			\hline
			\textbf{Tumor Sample} & \textbf{Proposed therapy} & \textbf{Dominated EG} \\
			\hline
			MM-0028-Tumor & Thalidomide, Pomalidomide, Dexamethasone & 26 \\
			MM-0038-Tumor & Thalidomide, Pomalidomide, Dexamethasone & 26 \\
			MM-0191-Tumor & Thalidomide, Pomalidomide, Dexamethasone & 26 \\
			\hline
		\end{tabular}
	\end{center}
\end{table}

\begin{table}
    \caption{The most effective standard drug therapy and the number of essential genes it dominates, for the minimum dominating set corresponding to the top ranked essential genes based on their eccentricity.}
	\label{table-drugs-m3ds-standardTherapy-eccentricity}
	\begin{center}
		\small
        \begin{tabular}{|l|l|c|}
            \hline
			\textbf{Tumor Sample} & \textbf{Standard therapy} & \textbf{Dominated EG} \\
			\hline
				MM-0028-Tumor & Bortezomib, Thalidomide, Dexamethasone & 26 \\
				MM-0038-Tumor & Bortezomib, Thalidomide, Dexamethasone & 26 \\
				MM-0191-Tumor & Bortezomib, Thalidomide, Dexamethasone & 26 \\
            \hline
        \end{tabular}
	\end{center}
\end{table}

\begin{table}
	\caption{The proposed drug therapy and the number of essential genes it dominates, for the minimum dominating set corresponding to the top ranked essential genes based on their betweenness.}
	\label{table-drugs-m3ds-betweenness}
	\begin{center}
		\small 
		\begin{tabular}{|l|l|c|}
			\hline
			\textbf{Tumor Sample} & \textbf{Proposed therapy} & \textbf{Dominated EG} \\
			\hline
			MM-0028-Tumor & Thalidomide, Pomalidomide, Dexamethasone & 27 \\
			MM-0038-Tumor & Thalidomide, Pomalidomide, Dexamethasone & 27 \\
			MM-0191-Tumor & Thalidomide, Pomalidomide, Dexamethasone & 27 \\
			\hline
		\end{tabular}
	\end{center}
\end{table}

\begin{table}
    \caption{The most effective standard drug therapy and the number of essential genes it dominates, for the minimum dominating set corresponding to the top ranked essential genes based on their betweenness.}
	\label{table-drugs-m3ds-standardTherapy-betweenness}
	\begin{center}
		\small
        \begin{tabular}{|l|l|c|}
            \hline
			\textbf{Tumor Sample} & \textbf{Standard therapy} & \textbf{Dominated EG} \\
			\hline
				MM-0028-Tumor & Bortezomib, Thalidomide, Dexamethasone & 27 \\
				MM-0038-Tumor & Bortezomib, Thalidomide, Dexamethasone & 27 \\
				MM-0191-Tumor & Bortezomib, Thalidomide, Dexamethasone & 27 \\
            \hline
        \end{tabular}
	\end{center}
\end{table}

\begin{table}
	\caption{The proposed drug therapy and the number of essential genes it dominates, for the minimum dominating set corresponding to the top ranked essential genes based on their eigenvector centrality.}
	\label{table-drugs-m3ds-eigenvector}
	\begin{center}
		\small 
		\begin{tabular}{|l|l|c|}
			\hline
			\textbf{Tumor Sample} & \textbf{Proposed therapy} & \textbf{Dominated EG} \\
			\hline
			MM-0028-Tumor & Thalidomide, Pomalidomide, Dexamethasone & 20 \\
			MM-0038-Tumor & Thalidomide, Pomalidomide, Dexamethasone & 21 \\
			MM-0191-Tumor & Thalidomide, Pomalidomide, Dexamethasone & 24 \\
			\hline
		\end{tabular}
	\end{center}
\end{table}

\begin{table}
    \caption{The most effective standard drug therapy and the number of essential genes it dominates, for the minimum dominating set corresponding to the top ranked essential genes based on their eigenvector centrality.}
	\label{table-drugs-m3ds-standardTherapy-eigenvector}
	\begin{center}
		\small
        \begin{tabular}{|l|l|c|}
            \hline
			\textbf{Tumor Sample} & \textbf{Standard therapy} & \textbf{Dominated EG} \\
			\hline
				MM-0028-Tumor & Bortezomib, Thalidomide, Dexamethasone & 20 \\
				MM-0038-Tumor & Bortezomib, Thalidomide, Dexamethasone & 21 \\
				MM-0191-Tumor & Bortezomib, Thalidomide, Dexamethasone & 24 \\
            \hline
        \end{tabular}
	\end{center}
\end{table}

\begin{table}
	\caption{The proposed drug therapy and the number of essential genes it dominates, for the minimum dominating set corresponding to the essential genes in the patient network.}
	\label{table-drugs-m3ds-all}
	\begin{center}
		\small 
		\begin{tabular}{|l|l|c|}
			\hline
			\textbf{Tumor Sample} & \textbf{Proposed therapy} & \textbf{Dominated EG} \\
			\hline
			MM-0028-Tumor & Thalidomide, Pomalidomide, Dexamethasone & 31 \\
			MM-0038-Tumor & Thalidomide, Pomalidomide, Dexamethasone & 31 \\
			MM-0191-Tumor & Thalidomide, Pomalidomide, Dexamethasone & 31 \\
			\hline
		\end{tabular}
	\end{center}
\end{table}

\begin{table}
    \caption{The most effective standard drug therapy and the number of essential genes it dominates, for the minimum dominating set corresponding to the essential genes in the patient network.}
	\label{table-drugs-m3ds-standardTherapy-all}
	\begin{center}
		\small
        \begin{tabular}{|l|l|c|}
            \hline
			\textbf{Tumor Sample} & \textbf{Standard therapy} & \textbf{Dominated EG} \\
			\hline
				MM-0028-Tumor & Bortezomib, Thalidomide, Dexamethasone & 31 \\
				MM-0038-Tumor & Bortezomib, Thalidomide, Dexamethasone & 31 \\
				MM-0191-Tumor & Bortezomib, Thalidomide, Dexamethasone & 31 \\
            \hline
        \end{tabular}
	\end{center}
\end{table}

From the point of view of centrality analysis, the results are consistent over all the networks and centrality measures with the prescription of the 3-drug combination of Thalidomide, Pomalidomide and Dexamethasone. The parallel outcomes yielded by this approach and the one focused on target controllability support our hypothesis of an interrelationship between the network topological properties and the genes influencing the disease. On these grounds, we propose hub genes ascertained by the centrality analysis of a disease network to be considered as basis for the discovery of new essential genes and drug targets.

%% file: section-conclusion.tex
\section{Conclusion}
\label{section-conclusion}

Network medicine is an exciting and promising field of research, with a high potential for personalized approaches. It brings together approaches in graph theory, network science, systems biology, bioinformatics and medicine, opening the door to detailed patient- and disease-specific insights. We discussed in this survey several basic methods of network modeling and their potential applicability in personalized medicine. We also demonstrated their potential on three multiple myeloma patient datasets. We showed how various methods (topological analysis, systems controllability) can be combined to predict optimal and personalized drug combinations therapies. In some cases, they differ from the initial standard therapy lines routinely offered to multiple myeloma patients, and resemble in part the options becoming available later in the disease progression. More studies (especially longitudinal studies) are needed to explore the full potential of these methods and convincingly demonstrate their applicability in the clinical practice. 

The results we obtained on the three datasets differ slightly from method to method. This is not surprising, as each method identifies different nodes and paths in the graphs that are of interest from various computational points of view. Which ones work best in practice should be explored with other research instruments, and it is likely the results may differ from case to case. 

The results in network medicine are critically dependent on the quality of the patient networks being applied to. The patient data that can be included in such networks can be quite diverse, including genetic mutations, copy number variations, differential gene expression,  co-morbidities, and concurrent treatments. All these data sources contribute to the set of nodes in the network. The interactions included in the network typically come from various interaction databases (some of which we discussed in this survey). The data going into these databases is very diverse: some of it is experimental (but not always on human samples), some is inferred from other experiments, while others are deduced through various machine learning methods. The importance of curating these datasets or at least choosing carefully which to include in the analyses cannot be underestimated. 